\documentclass{article}

\usepackage{amsmath,amsthm,amssymb,bm,color,mathrsfs,mathtools}
\usepackage{amscd,verbatim}
\usepackage{bbding,wasysym,pifont}
\usepackage[breaklinks=true,colorlinks=true,linkcolor=blue,urlcolor=blue,citecolor=blue]{hyperref}
\usepackage{graphicx}
\usepackage{multirow}
\usepackage{comment}
\usepackage[arrow,matrix,curve]{xy}
\usepackage{authblk}
\usepackage{tikz-cd}
\usepackage{subcaption}
\usepackage{colortbl}
\usepackage{rotating}

\usepackage{tikz-cd}
\usetikzlibrary{decorations.markings}
\tikzset{->-/.style={decoration={
  markings,
  mark=at position #1 with {\arrow{>}}},postaction={decorate}}}
\tikzset{->-/.default=0.5}

\newcommand\wt{\widetilde}

\newcommand{\CC}{{\mathbb C}}
\newcommand{\PP}{{\mathbb P}}
\newcommand{\RR}{{\mathbb R}}
\newcommand{\TT}{{\mathbb T}}
\newcommand{\ZZ}{{\mathbb Z}}

\newcommand{\NN}{{\mathbb N}}

\newcommand{\im}{\mathrm{i}}

\newtheorem{theorem}{Theorem}[section]

\newtheorem{lemma}[theorem]{Lemma}
\newtheorem{proposition}[theorem]{Proposition}

\newtheorem{definition}[theorem]{Definition}
\newtheorem{remark}[theorem]{Remark}

\long\def\red#1{\textcolor {red}{#1}} 
\long\def\blue#1{\textcolor {blue}{#1}}

\begin{document}

\title{Edge-following topological states}

\author{Guo Chuan Thiang 
%\thanks{Corresponding author: guochuan.thiang@adelaide.edu.au}
}
\affil{School of Mathematical Sciences, University of Adelaide, SA 5005, Australia}

\maketitle

\begin{abstract}
We prove that Chern insulators have topologically protected edge states which not only propagate unidirectionally along a straight line boundary, but also swerve around arbitrary-angled corners and geometric imperfections of the material boundary. This is a physical manifestation of the index theory of certain semigroup operator algebras.
\end{abstract}

\section{Introduction}
{\bf Physical motivation.} It is by now an old physics result that the quantised Hall conductance in the bulk (i.e.\ deep in the interior) of a 2D quantum Hall system leads to quantised boundary currents along its boundary, due to the appearance of topologically protected states near its boundary. The related \emph{Chern insulator}, recalled in Section \ref{sec:abstract.Chern.bulk}, has its roots in a model for the quantum anomalous Hall effect \cite{Haldane} which has no external magnetic field. It is a band insulator in the bulk, characterised by a certain topological invariant (a Chern number $k\in\ZZ$), at least abstractly. The physical expectation remains: there must appear $k$ chiral boundary states filling up the insulating energy gap. The abstract Chern insulator idea has quickly been exported from solid-state electron systems to many other areas of physics such as photonics \cite{Lu}, acoustics \cite{Ding}, cold atoms \cite{Jotzu}, gyroscopic metamaterials \cite{Nash}, Floquet systems \cite{Rechtsman}, mechanics \cite{Susstrunk}, exciton-polariton systems \cite{Klembt}. 

The adjective \emph{chiral} is intended to mean that the boundary states persist even when they encounter a \emph{corner} of the material, and indeed propagate \emph{around corners} and generally ``follow the edge'' without dissipation, even when bumps are present; see Fig.\ \ref{fig:bumpy.corner} for a schematic diagram. These remarkable properties have been seen experimentally, e.g. Fig.\ 3-5 of \cite{Rechtsman}, Fig.\ 3 of \cite{Lu}, Fig.\ 3 of \cite{Klembt}, Fig.\ 4 of \cite{Susstrunk}, Movies S1-3 of \cite{Nash}, and are nowadays considered hallmarks of topological protected edge states distinguishing them spurious ones. Following \cite{Haldane}, many popular models of Chern insulators are realised on a honeycomb lattice, which has $\ZZ^2$ translational symmetry generated by non-orthogonal vectors. Consequently, there are two basic types of boundary conditions --- armchair and zig-zag --- which are commonly studied (see Fig.\ \ref{fig:honeycomb}), and chiral boundary states have been observed to propagate around a corner where one condition switches to the other \cite{Klembt}.

{\bf Main result.} We prove that lattice models of Chern insulators, in any physical incarnation, have bulk-gap-filling spectra which give rise to \emph{quantised topologically protected boundary currents propagating along the material boundary, following it around corners and imperfections}. The ability to make precise computations for arbitrary imperfect corners is especially new. A concrete ``\emph{edge-travelling operator}'' $w_\urcorner$ (Remark \ref{rem:edge.travelling.op}, illustrated in Fig.\ \ref{fig:bumpy.corner}) is constructed, as a representative for the $K$-theory index of Chern insulator boundary states. This construction is fundamental for \emph{coarse index} computations in a follow-up paper \cite{Ludewig-Thiang-cobordism} addressing \emph{differential operator} models of Chern insulators and quantum Hall systems. 

{\bf Mathematical approach and previous work.} Topological boundary states can be derived in a $C^*$-algebraic approach \cite{KRS, PSB} (recalled in Section \ref{sec:half.plane}), building on work of Bellissard \cite{Bellissard} who introduced operator $K$-theory (see \cite{Rordam,WO} for gentle introductions) and noncommutative index theory methods into solid-state physics. In this formalism, which contains at its heart the index theory of Toeplitz operators, many important features of the above bulk-boundary correspondence can be accurately derived. Furthermore, although not discussed in this paper important physical effects due to disorder can be handled, and powerful abstract machinery from $KK$-theory can also be brought to bear.

However, these methods have so far been applied only to a limited geometric setup in which the material occupies a nice Euclidean half-plane, whereas in experimental practice, physicists are now especially interested in robustness of boundary states against changes in geometry of the boundary. To understand boundary states for more complicated material boundaries, the first step is to construct the analogue of the Toeplitz-like $C^*$-algebra extensions, which is straightforward, or at least prescriptive. The second step is to compute the $K$-theory groups and the connecting maps between them. The difficulty of this computational step may have been a pragmatic reason why very little has been done beyond the Euclidean half-plane setup (but see \cite{MTHyp} for an attempt for hyperbolic half-planes). Prior to this work, there had been no verification that the crucial exponential map (\S \ref{sec:boundary.states}) is always non-vanishing, regardless of boundary shape --- it may well be the case that the $K$-theory -- $C^*$-algebra machinery predicts trivial results except for perfect half-spaces, which would be damaging to the program. Let us remark that the $K$-theory groups for quarter-plane $C^*$-algebras already depend in a sensitive way on the angle of the corner (see \S \ref{sec:irrational.computations}), so it is not at all obvious that the exponential map survives such angle changes.

Fortuitously, much recent progress has been made in the $K$-theory of semigroup $C^*$-algebras \cite{Cuntz}, which generalise Toeplitz algebras in the appropriate way for our purposes. In particular, subsemigroups $S\subset \ZZ^2$ and their associated Toeplitz $C^*$-algebras $C^*_r(S)$, as initially studied by \cite{CoDo, DoHo, Ji, Jiang, ParkThesis, Parkcones, ParkSchochet}, are exactly the concept required for studying quarter-planes modelling material corners. We mention in passing that these works were recently utilised to demonstrate another interesting notion of ``bulk-edge to corner'' correspondence \cite{Hayashi, Hayashi2}. 

\emph{Relationship with coarse geometry.} For differential operator (``continuum'') models, quantum Hall Hamiltonians for uniform magnetic fields had been shown to have chiral edge states along boundaries with fairly general geometries \cite{FGW}, but quantisation was not established. In a companion paper \cite{Ludewig-Thiang-cobordism} to this one, the author proves that the boundary $K$-theory index of Chern insulators and quantum Hall Hamiltonians is a \emph{coarse index}, and utilises its cobordism invariance to prove that quantisation of the Chern insulator boundary current persists under deformations of the boundary preserving its coarse geometry. The methods and results of \cite{Ludewig-Thiang-cobordism} consistently complement those of this paper, in the sense that lattice models here are derived, in principle, from continuum ones via localised Wannier basis construction \cite{Ludewig-Thiang}.

{\bf Outline.} After recalling some background material, we construct in Section \ref{sec:quarter.plane.section} some generalisations of $C^*_r(S)$ as extensions of $C^*_r(\ZZ^2)$, which are the physical algebras of operators to which Hamiltonian operators on quarter-planes are affiliated. We are able to explicitly handle quarter-planes with any rational slope faces and arbitrary imperfections/bumps in a finite region, and in all cases define a topological invariant measuring the number of boundary-following modes. All relevant $K$-theory groups and connecting maps are computed in Section \ref{sec:quarter.plane.section} in terms of concrete generators, while cyclic cocycles computing the boundary currents are constructed in Section \ref{sec:topological.cornering.states} to demonstrate their quantisation. We also outline how this work extends to irrational slope cases, concave corners, and the quantum Hall effect with magnetic translations (Section \ref{sec:generalisations}).

\begin{figure}

\begin{center}
\begin{tikzpicture}

\draw (0,0) -- (-0.5,0.866) -- (0,1.732) -- (-0.5,2.598) -- (0,3.464);
\draw (3,0) -- (2.5,0.866) -- (3,1.732) -- (2.5,2.598) -- (3,3.464);
\draw (6,0) -- (5.5,0.866) -- (6,1.732) -- (5.5,2.598) -- (6,3.464);
\draw (1,0) -- (1.5,0.866) -- (1,1.732) -- (1.5,2.598) -- (1,3.464);
\draw (4,0) -- (4.5,0.866) -- (4,1.732) -- (4.5,2.598) -- (4,3.464);

\draw (0,0) -- (1,0);
\draw (3,0) -- (4,0);
\draw (0,1.732) -- (1,1.732);
\draw (3,1.732) -- (4,1.732);
\draw (0,3.464) -- (1,3.464);
\draw (3,3.464) -- (4,3.464);
\draw (1.5,0.866) -- (2.5,0.866);
\draw (4.5,0.866) -- (5.5,0.866);
\draw (1.5,2.598) -- (2.5,2.598);
\draw (4.5,2.598) -- (5.5,2.598);

\draw[ultra thick] (1,0) -- (1.5,0.866) -- (1,1.732) -- (1.5,2.598) -- (1,3.464);
\draw[ultra thick, dashed] (1,0) -- (1.5,0.866) -- (2.5,0.866) -- (3,0) -- (4,0) -- (4.5,0.866) -- (5.5,0.866) -- (6,0); 

\draw[->-] (1,0) -- (2.5,0.866);
\node[above left] at (2.5,0.4) {${\mathbf a}$};
\draw[->-] (1,0) -- (1,1.732);
\node[left] at (1.1,1.4) {${\mathbf b}$};

\fill[lightgray, opacity=0.2] (1,0) -- (2.5,0.866) -- (2.5,2.598) -- (1,1.732) -- (1,0);

\node at (1,0) {$\bullet$};
\node at (4,0) {$\bullet$};
\node at (1,1.732) {$\bullet$};
\node at (4,1.732) {$\bullet$};
\node at (1,3.464) {$\bullet$};
\node at (4,3.464) {$\bullet$};
\node at (-0.5,0.866) {$\bullet$};
\node at (-0.5,2.598) {$\bullet$};
\node at (2.5,0.866) {$\bullet$};
\node at (5.5,0.866) {$\bullet$};
\node at (2.5,2.598) {$\bullet$};
\node at (5.5,2.598) {$\bullet$};

\end{tikzpicture}

\end{center}
\caption{Honeycomb lattice, with vectors $\mathbf{a}$ and $\mathbf{b}$ generating the sublattice $\ZZ^2$ of translation symmetries. A fundamental domain is shaded. Translates of a vertex are marked with $\bullet$, while the unmarked ones are the $\ZZ^2$-translates of a second vertex. A vertical edge leads to the zig-zag boundary conditions (thick line), while a horizontal edge leads to armchair boundary conditions (dashed lines). The horizontal translation symmetry operation is $(2,-1)$ in terms of the basis $\mathbf{a}, \mathbf{b}$ for $\ZZ^2$.}\label{fig:honeycomb}
\end{figure}
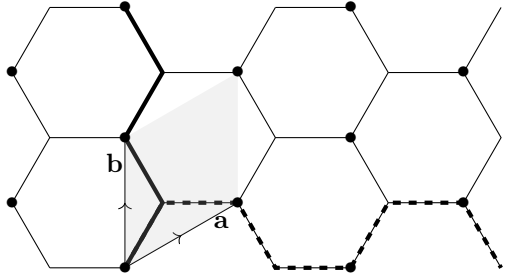

\section{Recap: chiral edge states of Chern insulators in half-plane geometry}
\subsection{Abstract Chern insulator}\label{sec:abstract.Chern.bulk}
Consider a lattice $\ZZ^2$ of translations acting on the Euclidean plane, identified with an orbit (the ``atomic sites'') after picking an origin. A tight-binding Hilbert space with two degrees of freedom per lattice site is $\mathscr{H}=\ell^2(\ZZ^2)\otimes \CC^2$, where $\ell^2(\ZZ^2)=\ell^2_{\rm reg}(\ZZ^2)$
denotes the (left or right) regular representation of $\ZZ^2$. Thus there is a unitary representation of $\gamma\in\ZZ^2$ as operators $U_{\gamma}\otimes 1$ acting by translation on the $\ell^2(\ZZ^2)$ factor and trivially on the ``internal'' factor $\CC^2$. Note that the operator-norm closure of the algebra generated by $U_\gamma, \gamma\in \ZZ^2$, is just the reduced group $C^*$-algebra $C^*_r(\ZZ^2)$, and is isomorphic to the algebra $C(\TT^2)$ of continuous functions on the Pontryagin dual $\TT^2$ of $\ZZ^2$ (\emph{Brillouin torus} in physics), by the Fourier transform. As generators for $C^*_r(\ZZ^2)$, we can take the basic horizontal translation $U_x=U_{(1,0)}$ and vertical translation $U_y=U_{(0,1)}$.

In what follows, we will suppress the subscript on $\ell^2_{\rm reg}(\ZZ^2)$. A \emph{tight-binding bulk Hamiltonian} is a bounded self-adjoint operator 
\begin{equation}
H=\sum_{\gamma\in\ZZ^2} U_\gamma\otimes W_\gamma\;\in \mathcal{B}(\ell^2(\ZZ^2)\otimes\CC^2),\label{eqn:tight.binding.Hamiltonian}
\end{equation}
where each $W_\gamma$ is a $2\times 2$ \emph{hopping matrix} satisfying $W_{\gamma}^*=W_{-\gamma}$ to ensure $H=H^*$. Note that $H$ is translation invariant. When the decay of $W_\gamma$ is sufficiently fast (e.g.\ only finitely many non-zero terms in the case of \emph{finite hopping range} Hamiltonians), the Fourier transform $\mathcal{F}:\ell^2(\ZZ^2)\otimes \CC^2\rightarrow L^2(\TT^2)\otimes \CC^2$ effects
$$ \mathcal{F}H\mathcal{F}^{-1}=\int_{\TT^2}^\oplus H_k\,dk$$
with $\TT^2\ni k\mapsto H_k$ a continuous (or even smooth) family of $2\times 2$ Hermitian matrices and $dk$ the normalised Haar measure. 

The spectrum $\sigma(H)$ of $H$ is the union of the $\sigma(H_k)$ over $k\in\TT^2$ . Suppose the spectrum of $H$ is $\sigma(H)=[a,b]\cup[c,d]$ with $b<c$ --- we call this the \emph{bulk spectral gap hypothesis} (see Fig.\ \ref{fig:gap.filling}). Then there is a continuous (or smooth) eigenspace assignment $k\mapsto \mathcal{L}^-_k\in\CC\PP^1$ where $\mathcal{L}^-_k$ is the lower energy eigenspace of $H_k$. This assignment can be thought of as the classifying map for the continuous (or smooth) \emph{valence} line bundle $\mathcal{L}^-\rightarrow\TT^2$ of eigenspaces for energies below the spectral gap. In this language, a \emph{Chern insulator} is a Hamiltonian $H$ whose valence bundle $\mathcal{L}^-$ is topologically non-trivial, as measured exactly by its non-vanishing first Chern class $c_1(\mathcal{L}^-)\in H^2(\TT^2,\ZZ)$, or equivalently, reduced $K$-theory class $[\mathcal{L}^-]-[1]\in\wt{K}^0(\TT^2)$. Specific choices of hopping matrices $W_y$ which result in nontrivial $\mathcal{L^-}$ are known, e.g.\ \cite{Haldane, PSB}.

{\bf Chern insulator in $C^*$-algebra language.} We have $H\in M_2(C^*_r(\ZZ^2))$ (the $2\times 2$ matrix algebra over $C^*_r(\ZZ^2)$), realised concretely on the Hilbert space $\mathscr{H}=\ell^2(\ZZ^2)\otimes\CC^2$. Rather than Fourier transforming and then constructing the valence bundle $\mathcal{L}^-$, we can directly construct by functional calculus the spectral projection $P_-=\varphi(H)\in M_2(C^*_r(\ZZ^2))$ onto energies below and including $b$: take $\varphi$ to be any continuous (or smooth) real-valued function which is 1 on the interval $(-\infty,b]$ and 0 on $[c,\infty)$. Then $P_-$ defines a $K$-theory class in $K_0(C^*_r(\ZZ^2))$. Now, the non-trivial generator of $K_0(C^*_r(\ZZ^2))\cong\ZZ^2$ can be taken to be the Bott projection\footnote{A general construction of $\mathfrak{b}\in M_2(C(\TT^2))$ with smooth entries can be found in \S2 of \cite{Loring}.} $\mathfrak{b}\in M_2(C^*_r(\ZZ^2))\cong M_2(C(\TT^2))$ (the other generator is the class $[1]$ of the identity element projection). After passing to cohomology via the Chern character, $\mathfrak{b}$ corresponds to the line bundle $\mathcal{L}\rightarrow \TT^2$ with Chern class the generator of $H^2(\TT^2,\ZZ)\cong\ZZ$. Thus a Chern insulator whose valence bundle has Chern class $k$ equivalently has 
\begin{equation}
[P_-]=(1-k)[1]+k[\mathfrak{b}].\label{eqn:projection.to.Bott}
\end{equation}

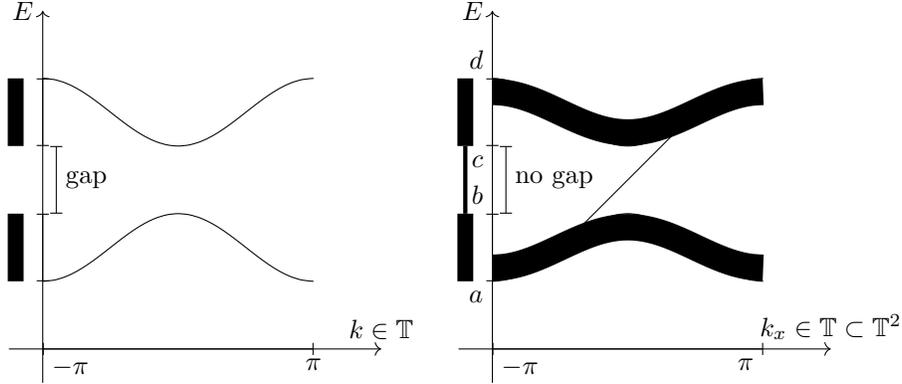
\begin{figure}

\begin{center}
\begin{tikzpicture}[scale=.9]
\draw[thin,->] (0,-0.5) -- (0,5);
\draw[thin,->] (-0.5,0) -- (5,0);
\draw[line width=6pt] (-0.4,1) -- (-0.4,2);
\draw[line width=6pt] (-0.4,3) -- (-0.4,4);

\draw[|-|] (0,1) -- (0,2);
\draw[|-|] (0,3) -- (0,4);
\draw[|-|] (0.2,2) -- (0.2,3);
\draw[|-|] (0,0) -- (4,0);

\node[right] at (0.2,2.5) {gap};
\node[below right] at (0,0) {$-\pi$};
\node[below] at (4,0) {$\pi$};
\node[left] at (0,5) {$E$};
\node[above] at (5,0) {$k\in\TT$};

\draw[domain=0:4,smooth,variable=\r] plot ({\r},{0.5*sin(deg(\r-1)*3.14*0.5)+1.5});
\draw[domain=0:4,smooth,variable=\r] plot ({\r},{-0.5*sin(deg(\r-1)*3.14*0.5)+3.5});
\end{tikzpicture}
\hspace{0.5em}
\begin{tikzpicture}[scale=.9]
\draw[thin,->] (0,-0.5) -- (0,5);
\draw[thin,->] (-0.5,0) -- (5,0);
\draw[line width=6pt] (-0.4,1) -- (-0.4,2);
\draw[line width=6pt] (-0.4,3) -- (-0.4,4);
\draw[ultra thick] (-0.4,2) -- (-0.4,3);
\draw (1,1.5) -- (3,3.5);
\draw[|-|] (0,1) -- (0,2);
\draw[|-|] (0,3) -- (0,4);
\draw[|-|] (0.2,2) -- (0.2,3);
\draw[|-|] (0,0) -- (4,0);
\node[below left] at (0,1) {$a$};
\node[above left] at (0,2) {$b$};
\node[below left] at (0,3) {$c$};
\node[above left] at (0,4) {$d$};
\node[right] at (0.2,2.5) {no gap};
\node[below right] at (0,0) {$-\pi$};
\node[below left] at (4,0) {$\pi$};
\node[left] at (0,5) {$E$};
\node[above] at (5,0) {$k_x\in\TT\subset\TT^2$};

\draw[domain=0:4,smooth,variable=\r] plot ({\r},{0.5*sin(deg(\r-1)*3.14*0.5)+1.5});
\draw[line width=10pt,domain=0:4,smooth,variable=\r] plot ({\r},{0.3*sin(deg(\r-1)*3.14*0.5)+1.5});
\draw[domain=0:4,smooth,variable=\r] plot ({\r},{-0.5*sin(deg(\r-1)*3.14*0.5)+3.5});
\draw[line width=10pt,domain=0:4,smooth,variable=\r] plot ({\r},{-0.3*sin(deg(\r-1)*3.14*0.5)+3.5});
\end{tikzpicture}

\end{center}
\caption{(L) Spectrum of a bulk Hamiltonian $H$ in one spatial dimension, with separated energy intervals. Its energy-momentum dispersion $E=E(k), k\in\TT$ is plotted. (R) Bulk Hamiltonian of a 2D Chern insulator with spectrum $[a,b]\cup[c,d]$ initially having a gap. Energy-momentum dispersion is indicated by thickened curved bands because only the dependence on one coordinate $k_x\in\TT^2$ is plotted while the dispersion in $k_y$ is collapsed. Half/quarter-plane truncations $\hat{H}$ of $H$ acquire new spectra (thin line) filling the spectral gap of $H$.}\label{fig:gap.filling}
\end{figure}

\subsection{Boundary topological invariants in half-plane geometry}\label{sec:half.plane}

While the bulk Hamiltonian $H$ for a Chern insulator has a spectral gap, the ``true'' Hamiltonian $\hat{H}$ (which is supposed to be $H$ acting on a restricted Hilbert space with appropriate boundary conditions) has extra ``chiral edge states'' filling up the spectral gap of $H$, which decay rapidly into the bulk and propagate unidirectionally along the material boundary, see Fig.\ \ref{fig:gap.filling}. 

The simplest way to model a material boundary is to truncate $\ell^2(\ZZ^2)$ to $\ell^2(\NN\times\ZZ)$. Thus the material occupies the right half-plane with straight line boundary $x=0$. We give a brief outline of how the language of Toeplitz extensions \cite{KRS, PSB} is used to prove the existence of chiral edge states and their finer analytic properties, with much more detail available in \cite{PSB}.

\subsubsection{Index theory of classical Toeplitz operators}\label{sec:classical.Toeplitz}
The following is classical, e.g. \S 3.F of \cite{WO}. With $U_x$ the unitary shift operator on $\ell^2(\ZZ)$, let $\hat{U}_x$ denote its truncation to the (right) unilateral shift operator on the ``right-half line'' Hilbert space $\ell^2(\NN)$ (after a Fourier transform, $\ell^2(\NN)$ is the classical Hardy subspace of $L^2(\TT)$). The \emph{Toeplitz} algebra $\mathcal{T}$ is the $C^*$-subalgebra of $\mathcal{B}(\ell^2(\NN))$ generated by $\hat{U}_x$. There is a short exact sequence
$$0\rightarrow \mathcal{K}(\ell^2(\NN))\rightarrow \mathcal{T}\xrightarrow{\pi}C^*_r(\ZZ)\overset{\rm Fourier}{\cong} C(\TT)\rightarrow 0,$$
where $\pi$ is the \emph{symbol} map taking $\hat{U}_x\mapsto U_x$. An Toeplitz operator $T_f\in\mathcal{T}$ with invertible symbol function $f\in C(\TT)\cong C^*_r(\ZZ)$ is Fredholm with \emph{analytic} index equal to minus the winding number of $f$ (the \emph{topological} index). This integer is invariant under $T_f+K$ with $K$ compact.
Correspondingly, the $K$-theory six-term long exact sequence
\begin{equation*}
\xymatrix{
{\overbrace{K_0(\mathcal{K})}^{\ZZ}} \ar[r] & K_0(\mathcal{T}) \ar[r] & {\overbrace{K_0(C(\TT))}^{\ZZ}} \ar[d]^{{\rm Exp}} \\
{\underbrace{K_1(C(\TT))}_{\ZZ}} \ar[u]^{{\rm Ind}} & K_1(\mathcal{T}) \ar[l] & {\underbrace{K_1(\mathcal{K})}_{0}}\ar[l]_{\quad 0}
}
\end{equation*}
has connecting index map an isomorphism, and so $K_0(\mathcal{T})\cong \ZZ$, $K_1(\mathcal{T})=0$. 

We remark that even though the other connecting map ${\rm Exp}$ is the trivial map in the above case, its \emph{existence} is a deep result drawing upon Bott periodicity. We will exploit its utility in detecting obstructions to difficult lifting problems.
%(see \S9.3.2 of \cite{Blackadar})

\subsubsection{Toeplitz extension for half-plane geometry}\label{sec:half.plane.Toeplitz}
The two-dimensional version of the Toeplitz ``half-line algebra'' $\mathcal{T}$ is a ``half-plane algebra'', which will contain the half-plane versions of $H$. Let $\iota$ be the inclusion $\ell^2(\NN\times\ZZ)\rightarrow \ell^2(\ZZ^2)$, and $p:\ell^2(\ZZ^2)\rightarrow \ell^2(\NN\times\ZZ)$ its adjoint orthogonal projection. Each operator $A\in C^*_r(\ZZ^2)\subset \mathcal{B}(\ell^2(\ZZ^2))$ has a truncation $\hat{A}\coloneqq p\circ A\circ\iota \in \mathcal{B}(\ell^2(\NN\times\ZZ))$, and this assignment is $*$-linear but \emph{not} multiplicative --- for example, $\hat{U}_x$ is a \emph{non-unitary} isometry satisfying $\hat{U}_x\hat{U}_x^*=1-P_{x=0}$ where $P_{x=0}$ is the orthogonal projection onto the ``boundary subspace'' of $\ell^2(\ZZ^2)$ spanned by basis vectors at $(0,n), n\in\ZZ$. Let $C^*_r(\NN\times\ZZ)$ be the $C^*$-algebra generated by the $\hat{U}_\gamma, \gamma\in\NN\times\ZZ$ (or equivalently $\gamma\in\ZZ^2$ due to $\hat{U}_\gamma^*=\hat{U}_{\gamma^{-1}}$). The isometry $\hat{U}_x$ and the unitary operator $\hat{U}_y$ are already enough to generate $C^*_r(\NN\times\ZZ)$, which exhibits the isomorphism $C^*_r(\NN\times\ZZ)\cong \mathcal{T}\otimes C^*_r(\ZZ)$.  There is a short exact sequence
\begin{equation}
0\rightarrow\mathcal{J} \rightarrow C^*_r(\NN\times\ZZ) \xrightarrow{\pi} C^*_r(\ZZ^2)\rightarrow 0,\label{eqn:SES.standard.half.plane}
\end{equation}
where $\pi$ is the $*$-homomorphism defined on generators by $\hat{U}_x\mapsto U_x, \hat{U}_y\mapsto U_y$. The kernel $\mathcal{J}$ is in fact the commutator ideal in $C^*_r(\NN\times\ZZ)$, cf.\ \cite{CoDo}. Since $[\hat{U}_y,\hat{U}_x^*]=0$, we see that $\mathcal{J}$ is just the ideal generated by $P_{x=0}=1-\hat{U}_x\hat{U}_x^*$, and we have $\mathcal{J}\cong\mathcal{K}(\ell^2(\NN))\otimes C^*_r(\ZZ)$.

For any $A\in C^*_r(\ZZ^2)$, the truncations $\hat{A}=p\circ A\circ \iota$ are elements of $C^*_r(\NN\times\ZZ)$. If $A\in M_N(C^*_r(\ZZ^2))$, there is a similar truncation to $\hat{A}\in M_N(C^*_r(\NN\times\ZZ))$, so for example, we can truncate the bulk Hamiltonian $H$ to a ``half-plane Hamiltonian'' $\hat{H}$ acting on $\ell^2(\NN\times\ZZ)\otimes\CC^2$. The homomorphism $\pi$ extends to matrix algebras, and (retaining the same notation) $\pi(\hat{H})=H$. We could also consider more general half-plane Hamiltonians $\hat{H}'=\hat{H}+\tilde{H}$ with $\tilde{H}=\tilde{H}^*\in\mathcal{J}$ some extra ``boundary Hamiltonian'' term. We still have $\pi(\hat{H}')=H=\pi(\hat{H})$, so that the spectrum of half-plane versions of $H$ generally contains that of $H$.

\subsubsection{Topological boundary states in half-plane geometry}\label{sec:boundary.states}
The extra spectra of the half-plane Hamiltonian have interesting features, deducible from the long exact sequence (LES) in $K$-theory for the sequence Eq.\ \eqref{eqn:SES.standard.half.plane}, cf.\ \S 4.3.1 of \cite{PSB}. This LES is easily computed (e.g.\ with K\"{u}nneth formula) to be
\begin{equation*}%\label{standard.half.plane.LES}
\xymatrix{
{\overbrace{K_0(\mathcal{J})}^{\ZZ}} \ar[r]^{0\quad} & {\overbrace{K_0(C^*_r(\NN\times\ZZ))}^{\ZZ[1]}} \ar[r] & {\overbrace{K_0(C^*_r(\ZZ^2))}^{\ZZ[1]\oplus\ZZ[\mathfrak{b}]}} \ar[d]^{{\rm Exp}} \\
{\underbrace{K_1(C^*_r(\ZZ^2))}_{\ZZ^2}} \ar[u]^{{\rm Ind}} & {\underbrace{K_1(C^*_r(\NN\times\ZZ))}_{\ZZ}} \ar[l] & {\underbrace{K_1(\mathcal{J})}_{\ZZ[U_y]}}\ar[l]_{\quad 0}
}
\end{equation*}
where the operator $\mathcal{J}\ni \hat{U}_yP_{x=0}\longleftrightarrow P_{x=0}\otimes U_y\in\mathcal{K}(\ell^2(\NN))\otimes C^*_r(\ZZ)$ effecting ``translation-along-the-boundary'' represents the $K_1(\mathcal{J})\cong\ZZ$ generator (see Fig.\ \ref{fig:standard.windings}). Actually $\mathcal{J}$ is non-unital, so its $K_1$-group representatives should be unitaries in the (matrix algebras over the) unitisation of $\mathcal{J}$ rather than $\mathcal{J}$ itself. This means that we implicitly regard $\hat{U}_yP_{x=0}$ as a unitary operator on $\ell^2(\NN\times\ZZ)$ by extending it by the identity operator on the complement of $x=0$. Notwithstanding this technicality, we abuse notation and simply write $K_1(\mathcal{J})\cong\ZZ[U_y]$. 

{\bf $K$-theory exponential map.} The connecting map $\rm{Exp}$ takes $[\mathfrak{b}]\mapsto -[U_y]$. Quite generally, ${\rm Exp}$ is a suspended/``higher'' index map, which measures the obstruction to lifting a projection $P$ in $M_N(C^*_r(\ZZ^2))$ to a projection in \mbox{$M_N(C^*_r(\NN\times\ZZ))$}, where $N$ is taken to be arbitrarily large. Namely, ${\rm Exp}[P]=[{\rm exp}(-2\pi i \hat{P})]$ where $\hat{P}$ is any self-adjoint lift of $P$, e.g.\ see \S 12.2 of \cite{Rordam}. So if ${\rm Exp}[P]\neq[1]$ (note that $[1]=0$ in $K_1$-theory), it cannot be the case that $\hat{P}$ is a projection.

Just as we have the spectral projection $P_-=\varphi(H)\in M_N(C^*_r(\ZZ^2))$ for a bulk Hamiltonian $H$ with spectral gap, we can similarly consider $\varphi(\hat{H})\in M_N(C^*_r(\NN\times\ZZ))$. Since $\pi(\varphi(\hat{H}))=\varphi(\pi(\hat{H}))=\varphi(H)=P_-$, we see that $\varphi(\hat{H})$ is a lift of $P_-$ in $M_N(C^*_r(\NN\times\ZZ))$, \emph{but this lift need not be a projection}. If $H$ is a Chern insulator, i.e.\ $[P_-]=(1-k)[1]+k[\mathfrak{b}]$ with $k\neq 0$ (Eq.\ \eqref{eqn:projection.to.Bott}), then ${\rm Exp}([P_-])=-k[U_y]\neq [1]$, so that no lift of $P_-$ is a projection. In particular, $\varphi(\hat{H})$ is not a projection, which means $\varphi(\lambda)\neq 0,1$ for some $\lambda$ in the spectrum $\sigma(\hat{H})$ of $\hat{H}$. We could have chosen $\varphi$ such that $\varphi(x)\neq 0,1$ only on some arbitrarily given small interval in the spectral gap $(b,c)$ of $H$. Then $\hat{H}$ must have some spectrum in this small interval, and by varying the choice of interval in $(b,c)$, we conclude that $\hat{H}$ actually has spectra filling the entire gap $(b,c)$. The same argument holds even if $\hat{H}$ is replaced by $\hat{H}'=\hat{H}+\tilde{H}$ with $\tilde{H}=\tilde{H}^*\in M_N(\mathcal{J})$ a boundary perturbation. The gap-filling spectra are furthermore exponentially localised near the boundary, cf.\ \S2.4.3 of \cite{PSB}.

For Chern insulators, ${\rm Exp}[P_-]\neq 0$ gives a non-vanishing ``boundary topological invariant''. By pairing ${\rm Exp}[P_-]$ with a suitable cyclic 1-cocycle, one obtains a local formula for the winding number $-k\in\ZZ$ of ${\rm Exp}[P_-]=-k[\hat{U}_y]$, which has the physical meaning of an integer quantised boundary current, see \S 7 of \cite{PSB}. We generalise this construction in Section \ref{sec:topological.cornering.states} for quarter-plane geometries.

\begin{remark}
The $K$-theory exponential map had also been used to obtain an index theorem for Dirac-type operators on noncompact partitioned manifolds (Theorem 3.3 of \cite{RoeToeplitz}), with the indices furthermore insensitive to certain deformations of the partitioning hypersurface (see Prop.\ 1.7, Theorem 7.7 of \cite{RoeToeplitz} for details). The connection of our work to this flavour of index theory is left for a future work \cite{Ludewig-Thiang-cobordism}.
\end{remark}

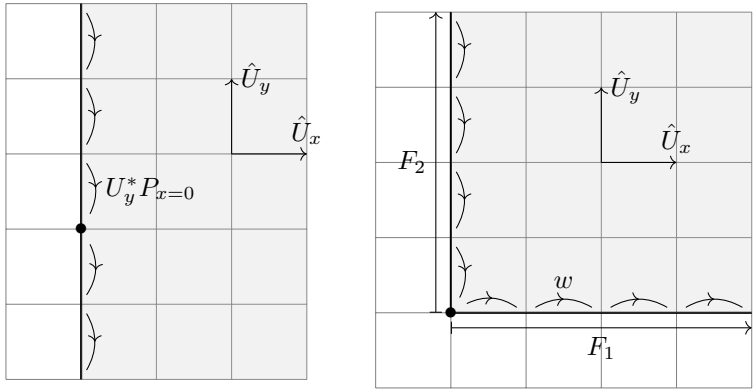
\begin{figure}

\begin{center}
\begin{tikzpicture}
\draw[help lines] (-1,-2) grid (3,3);
\draw[thick] (0,-2) -- (0,3);

\fill[lightgray, opacity=0.2] (0,-2) rectangle (3,3);

\node (a) at (0,3) {};
\node (b) at (0,2) {};
\node (c) at (0,1) {};
\node (d) at (0,0) {$\bullet$};
\node (e) at (0,-1) {};
\node (f) at (0,-2) {};

\path
    (a) edge[bend left, ->-] node [right] {} (b)
    (b) edge[bend left, ->-] node [right] {} (c)
    (c) edge[bend left, ->-] node [right] {} (d)
    (d) edge[bend left, ->-] node [right] {} (e)
    (e) edge[bend left, ->-] node [right] {} (f);

\node [right] at (0.2,0.5) {$U_y^*P_{x=0}$};

\node[above] at (3,1) {$\hat{U}_x$};
\node[right] at (2,2) {$\hat{U}_y$};
\draw [->] (2,1) -- (3,1);
\draw [->] (2,1) -- (2,2);
\end{tikzpicture}
\hspace{1em}
\begin{tikzpicture}
\draw[help lines] (-1,-1) grid (4,4);
\draw[thick] (0,4) -- (0,0) -- (4,0);
\draw [|->] (0,-0.2) -- (4,-0.2);
\draw [|->] (-0.2,0) -- (-0.2,4);

\fill[lightgray, opacity=0.2] (0,0) rectangle (4,4);

\node (a) at (0,4) {};
\node (b) at (0,3) {};
\node (c) at (0,2) {};
\node (d) at (0,1) {};
\node (e) at (0,0) {$\bullet$};
\node (f) at (1,0) {};
\node (g) at (2,0) {};
\node (h) at (3,0) {};
\node (i) at (4,0) {};

\path
    (a) edge[bend left, ->-] node [right] {} (b)
    (b) edge[bend left, ->-] node [right] {} (c)
    (c) edge[bend left, ->-] node [right] {} (d)
    (d) edge[bend left, ->-] node [right] {} (e)
    (e) edge[bend left, ->-] node [right] {} (f)
    (f) edge[bend left, ->-] node [right] {} (g)
    (g) edge[bend left, ->-] node [right] {} (h)
    (h) edge[bend left, ->-] node [right] {} (i);

\node [above] at (1.5,0.2) {$w$};

\node [left] at (-0.2,2) {$F_2$};
\node [below] at (2,-0.2) {$F_1$};
\node[above] at (3,2) {$\hat{U}_x$};
\node[right] at (2,3) {$\hat{U}_y$};
\draw [->] (2,2) -- (3,2);
\draw [->] (2,2) -- (2,3);
\end{tikzpicture}

\end{center}
\caption{Standard half-plane and quarter-plane geometries. The effect of the generator $[U_yP_{x=0}]$ of $K_1(\mathcal{J})$ and $[w]$ of $K_1(\mathcal{I}')$ are illustrated with curved arrows.}\label{fig:standard.windings}
\end{figure}

\section{$K$-theory of bumpy quarter-plane algebras}\label{sec:quarter.plane.section}
In this Section, we will generalise the machinery outlined in Section \ref{sec:half.plane}, to the ``bumpy quarter-plane'' setting.

\subsection{Standard quarter-plane Toeplitz algebra}\label{sec:standard.quarter.plane}
Instead of a half-plane, now consider the material occupying the \emph{standard} quarter-plane $C$, which is the upper-right quadrant bounded by the $x=0$ and $y=0$ axes. The set of lattice points in the material is $C\cap\ZZ^2$, which is the subsemigroup $\NN^2=\{(m,n)\,:\,m,n\in\NN\}\subset \ZZ^2$, see Fig. \ref{fig:standard.windings}. The half-plane Toeplitz algebra $C^*_r(\NN\times\ZZ)\cong \mathcal{T}\otimes C^*_r(\ZZ)$, has a generalisation to a ``quarter-plane Toeplitz algebra'' $C^*_r(\NN^2)$ as follows.

In the same vein as Section \ref{sec:half.plane.Toeplitz}, let $\iota:\ell^2(\NN^2)\hookrightarrow \ell^2(\ZZ^2)$ be the inclusion, with adjoint the orthogonal projection $p: \ell^2(\ZZ^2)\rightarrow \ell^2(\NN^2)$, so that $p\circ\iota=1_{\ell^2(\NN^2)}$. Given an operator $A\in C^*_r(\ZZ^2)$, its truncation to $\mathcal{B}(\ell^2(\NN^2))$ is $\hat{A}:=p\circ A\circ\iota$. This is a continuous $*$-linear assignment, although it does not respect products. As particular examples, for each unitary translation $U_\gamma, \gamma\in\ZZ$, we write $\hat{U}_\gamma=p\circ U_\gamma\circ \iota$. We also write $\hat{U}_x=p\circ U_x\circ \iota$ and $\hat{U}_y=p\circ U_y\circ \iota$, which are both non-unitary isometries.

For any subset $Y\subset\NN^2$, we denote the orthogonal projection $\ell^2(\NN^2)\rightarrow\ell^2(Y)$ by $P_Y$. Such projections commute among themselves.
Notice that for $\gamma\in\NN^2$, the truncation $\hat{U}_\gamma$ remains an isometry, $\hat{U}_\gamma^*\hat{U}_\gamma=1$, but its range projection $\hat{U}_\gamma \hat{U}_\gamma^*$ is generally not the identity. For example, $[\hat{U}_y^*,\hat{U}_y]=P_{F_1}$, the orthogonal projection onto $\ell^2(F_1)$ where $F_1=\{(m,0):m\in\NN\}$ is the ``horizontal boundary face'', see Fig.\ \ref{fig:standard.windings}. Similarly, $[\hat{U}_x^*,\hat{U}_x]=P_{F_2}$ where $F_2=\{(0,n):n\in\NN\}$ is the ``vertical boundary face''. We also have $0=[\hat{U}_x,\hat{U}_y]=[\hat{U}_x^*,\hat{U}_y^*]=[\hat{U}_x^*,\hat{U}_y]$.

\begin{definition}\label{defn:standard.quarter.plane}
The standard quarter-plane algebra, denoted $C^*_r(\NN^2)$, is defined to be the (unital) $C^*$-subalgebra of $\mathcal{B}(\ell^2(\NN^2))$ generated by $\hat{U}_x$ and $\hat{U}_y$.
\end{definition}
\noindent
It is easy to see that we can also define $C^*_r(\NN^2)$ to be the $C^*$-subalgebra generated by $\hat{U}_\gamma, \gamma\in\NN^2$, or by  $\hat{U}_\gamma, \gamma\in\ZZ^2$.

Let $\mathcal{I}'$ be the commutator ideal in $C^*_r(\NN^2)$. By inspecting the above commutators amongst $\hat{U}_x,\hat{U}_y$ and their adjoints, we may deduce that $\mathcal{I}'=\mathcal{I}'_1+\mathcal{I}'_2$ where for $i=1,2$, $\mathcal{I}'_i$ denotes the closed ideal generated by the projection $P_{F_i}$. We remark that $P_{F_1}P_{F_2}=P_{\{(0,0)\}}\in\mathcal{I}'$ is a rank-1 projection, so that in fact $\mathcal{K}(\ell^2(\NN^2))\subset\mathcal{I}'$. Modulo $\mathcal{I}'$, each $\hat{U}_\gamma, \gamma\in\NN^2$ is unitary, and the map $\pi:\hat{U}_\gamma\mapsto U_\gamma, \gamma\in\NN^2$ extends to a $*$-homomorphism $\pi:C^*_r(\NN^2)\rightarrow C^*_r(\ZZ^2)$ with kernel $\mathcal{I}'$, i.e.\ there is a short exact sequence, cf.\ \S 4 of \cite{CoDo},
\begin{equation}
0\rightarrow \mathcal{I}'\rightarrow C^*_r(\NN^2)\overset{\pi}{\rightarrow} C^*_r(\ZZ^2)\rightarrow 0. \label{standard.quarter.plane.SES}
\end{equation}
\noindent
The prime on the ideal $\mathcal{I}'$ is meant to distinguish it from the corresponding ideals $\mathcal{I}$ constructed in Eq.\ \eqref{staircase.SES} for general quarter-plane geometries later on.

We would like to compute the corresponding $K$-theory long exact sequence for Eq.\ \eqref{standard.quarter.plane.SES}. For this, we appeal to Prop.\ 1 of \cite{DoHo}, which says that there is a canonical isomorphism $C^*_r(\NN^2)\cong C^*_r(\NN)\otimes C^*_r(\NN)$, where each $C^*_r(\NN)$ factor is the classical Toeplitz algebra $\mathcal{T}$ (Section \ref{sec:classical.Toeplitz}). We saw that $K_0(C^*_r(\NN))\cong \ZZ$ and $K_1(C^*_r(\NN))=0$, so the K\"{u}nneth theorem gives $K_0(C^*_r(\NN^2))=K_0(C^*_r(\NN))\otimes_\ZZ K_0(C^*_r(\NN))\cong\ZZ$ and $K_1(C^*_r(\NN^2))=0$. We also know that $K_0(C^*_r(\ZZ^2))\cong\ZZ[1]\oplus\ZZ[\mathfrak{b}]$. Then the LES is

\begin{equation}\label{standard.quarter.plane.LES}
\xymatrix{
{\overbrace{K_0(\mathcal{I}')}^{\ZZ^2}} \ar[r]^{0} & {\overbrace{K_0(C^*_r(\NN^2))}^{\ZZ[1]}} \ar[r]^{\pi_*} & {\overbrace{K_0(C^*_r(\ZZ^2))}^{\ZZ[1]\oplus\ZZ[\mathfrak{b}]}} \ar[d]^{{\rm Exp}} \\
{\underbrace{K_1(C^*_r(\ZZ^2))}_{\ZZ^2}} \ar[u]^{{\rm Ind}}_{\cong} & {\underbrace{K_1(C^*_r(\NN^2))}_{0}} \ar[l]_{0} & {\underbrace{K_1(\mathcal{I}')}_{\ZZ[w]}}\ar[l]_{0}
}
\end{equation}
The exponential map must take $[\mathfrak{b}]$ to the generator of $K_1(\mathcal{I}')\cong\ZZ$. We may represent this generator by the operator $w=\hat{U}_y^*P_{F_2}+\hat{U}_xP_{F_1}$ translating anticlockwise along the boundary, as illustrated in Fig.\ \ref{fig:standard.windings}. Also, the generators of $K_0(\mathcal{I}')\cong\ZZ^2$ can be taken to be $[P_{F_1}]$ and $[P_{F_2}]$. These results follow from a direct calculation identical to that provided in Section \ref{sec:staircase} with the ``${}_\urcorner$'' decorations dropped (see Eq.\ \eqref{eqn:staircase.K0.ideal.generators}, \eqref{eqn:staircase.cornering}), or from \cite{Cuntz} as recalled in Section \ref{sec:quarter.plane.semigroup}.

\subsection{Standard quarter-plane with imperfect boundaries}\label{sec:staircase}
Let us modify the standard quarter-plane geometry with a ``staircase-shaped'' boundary condition, with a step corner at $(1,1)$, as illustrated in Fig. \ref{fig:staircase}. Thus, we need to truncate to $\ell^2(\NN^2_\urcorner)$ where $\NN^2_\urcorner:=\NN^2\setminus \{(0,0)\}$. We will construct the analogue of Eq.\ \eqref{standard.quarter.plane.SES} for this geometry.

\begin{figure}
\begin{center}
\begin{tikzpicture}
\draw[help lines] (0,0) grid (4,4);
\draw[thick] (0,4) --(0,1) -- (1,1) -- (1,0) -- (4,0);

\fill[lightgray, opacity=0.2] (0,4) --(0,1) -- (1,1) -- (1,0) -- (4,0) -- (4,4) -- (0,4);

\node[right] at (3,2) {$V_x$};
\node[above left] at (2,3) {$V_y$};
\draw[->] (2,2) -- (3,2);
\draw[->] (2,2) -- (2,3);

\node[left, above] at (0,0) {(0,0)};
\node at (0,0) {$\bullet$};
\node[red] at (1,0) {$\bullet$};
\node[red] at (2,0) {$\bullet$};
\node[red] at (3,0) {$\bullet$};
\node[red] at (4,0) {$\bullet$};
\node[red] at (0,1) {$\bullet$};

\node[blue] at (0,1) {$\times$};
\node[blue] at (0,2) {$\times$};
\node[blue] at (0,3) {$\times$};
\node[blue] at (0,4) {$\times$};
\node[blue] at (1,0) {$\times$};

\node[green] at (1,1) {$\circ$};

\node (a) at (0,4) {};
\node (b) at (0,3) {};
\node (c) at (0,2) {};
\node (d) at (0,1) {};
\node (e) at (1,1) {};
\node (f) at (1,0) {};
\node (g) at (2,0) {};
\node (h) at (3,0) {};
\node (i) at (4,0) {};

\path
    (a) edge[bend left, ->-] node [right] {} (b)
    (b) edge[bend left, ->-] node [right] {} (c)
    (c) edge[bend left, ->-] node [right] {} (d)
    (d) edge[bend left, ->-] node [right] {} (e)
    (e) edge[bend left, ->-] node [right] {} (f)
    (f) edge[bend left, ->-] node [right] {} (g)
    (g) edge[bend left, ->-] node [right] {} (h)
    (h) edge[bend left, ->-] node [right] {} (i);
    
 \node[above] at (2.5,0.2) {$w_\urcorner$};

\end{tikzpicture}
\hspace{1em}
\begin{tikzpicture}
\draw[help lines] (0,0) grid (4,4);
\draw[thick] (0,4) --(0,1) -- (1,1) -- (1,0) -- (4,0);

\fill[lightgray, opacity=0.2] (0,4) --(0,1) -- (1,1) -- (1,0) -- (4,0) -- (4,4) -- (0,4);

\node[left, above] at (0,0) {(0,0)};
\node at (0,0) {$\bullet$};

\node (a) at (1,4) {};
\node (b) at (1,3) {};
\node (c) at (1,2) {};
\node (d) at (1,1) {};
\node (e) at (2,1) {};
\node (f) at (3,1) {};
\node (g) at (4,1) {};

\path
    (a) edge[bend left, ->-] node [right] {} (b)
    (b) edge[bend left, ->-] node [right] {} (c)
    (c) edge[bend left, ->-] node [right] {} (d)
    (d) edge[bend left, ->-] node [right] {} (e)
    (e) edge[bend left, ->-] node [right] {} (f)
    (f) edge[bend left, ->-] node [right] {} (g);

 \node[above] at (2.5,1.2) {$w_\urcorner'$};

\end{tikzpicture}
\end{center}
\caption{Quarter-plane with a ``staircase'' boundary modification at the corner. The faces $F_1, F_2$ are indicated with the red $\red{\bullet}$ and the blue $\blue{\times}$ respectively. Two possible representatives $w_\urcorner, w_\urcorner'$ of the generator of $K_1(\mathcal{I}^\urcorner)$ are indicated.}\label{fig:staircase}
\end{figure}
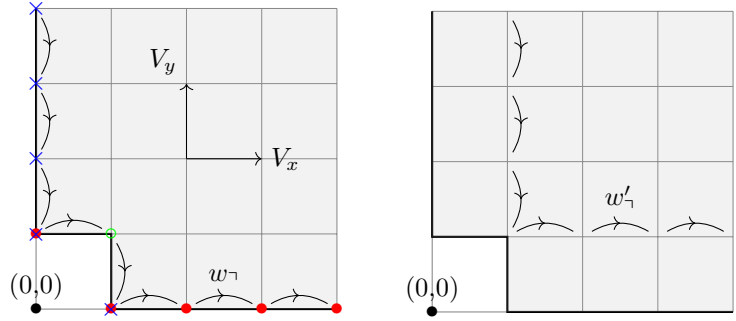

\begin{definition}\label{defn:staircase.algebra}
With $\iota:\ell^2(\NN^2_\urcorner)\rightarrow\ell^2(\ZZ^2)$ the inclusion and $p:\ell^2(\ZZ^2)\rightarrow\ell^2(\NN^2_\urcorner)$ its adjoint projection, let $V_\gamma=p\circ U_\gamma\circ \iota$ be the corresponding truncation of $U_\gamma, \gamma\in\ZZ^2$. We define $C^*_r(\NN^2_\urcorner)\subset \mathcal{B}(\ell^2(\NN^2_\urcorner))$ to be the unital $C^*$-algebra generated by $V_\gamma, \gamma\in\ZZ^2$, or equivalently by the isometries $V_\gamma, \gamma\in\NN^2$. 
\end{definition}

It is also enough to take as a generating set the basic isometries $V_x=V_{(0,1)}$ and $V_y=V_{(1,0)}$. Although $[V_x,V_y]=0=[V_y^*,V_x^*]$, there are nonvanishing commutators
\begin{equation}
[V_x^*,V_x]=P_{F_2^\urcorner},\qquad  [V_y^*,V_y]=P_{F_1^\urcorner},\qquad [V_y^*,V_x]=|(1,0)\rangle\langle(0,1)|,\label{commutators}
\end{equation}
where the ``faces'' $F_1^\urcorner, F_2^\urcorner$ (see Fig.\ \ref{fig:staircase}) are now defined by
$$F_1^\urcorner=\{(0,1)\}\cup \{(m,0):m\geq 1\},\qquad F_2^\urcorner=\{(1,0)\}\cup \{(0,n):n\geq 1\}.$$
Note that $V_y^*P_{F_1^\urcorner}=0=V_x^*P_{F_2^\urcorner}$ --- a face is ``killed'' when translated outwards.

\begin{lemma}\label{lem:staircase.commutator.ideal}
The commutators in Eq.\ \eqref{commutators} generate the (closed) commutator ideal $\mathcal{I}^\urcorner$ in $C^*_r(\NN^2_\urcorner)$, and the quotient is isomorphic to $C^*_r(\ZZ^2)$.
Thus there is a short exact sequence
\begin{equation}
0\rightarrow \mathcal{I}^\urcorner\rightarrow C^*_r(\NN^2_\urcorner)\overset{\pi_\urcorner}{\rightarrow} C^*_r(\ZZ^2)\rightarrow 0,\label{staircase.SES}
\end{equation}
where the quotient map $\pi_\urcorner$ takes $V_x\mapsto U_x$ and $V_y\mapsto U_y$.
\end{lemma}
\begin{proof}
Each product of the generators $V_x,V_x^*,V_y,V_y^*$ can be ``normal-ordered'', at the expense of some commutator-terms, into the form $(V_x^*)^aV_x^b(V_y^*)^cV_y^d$ for some (unique) non-negative integers $a,b,c,d$. Then after modding out the commutator terms, we can identify $(V_x^*)^aV_x^b(V_y^*)^cV_y^d\leftrightarrow U_x^{b-a}U_y^{d-c}\in C^*_r(\ZZ^2)$.
\end{proof}
\noindent
We wish to compute the $K$-theory 6-term exact sequence for Eq.\ \eqref{staircase.SES}.

 Let $\mathcal{I}_1^\urcorner$ and $\mathcal{I}_2^\urcorner$ be the ideals in $C^*_r(\NN^2_\urcorner)$ generated by the commutators $[V_y^*,V_y]=P_{F_1^\urcorner}$ and $[V_x^*,V_x]=P_{F_2^\urcorner}$ respectively. 
 
 \begin{lemma}\label{lem:commutator.ideal}
 The commutator ideal $\mathcal{I}^\urcorner={\rm ker}(\pi_\urcorner)$ can be written as
 $$\mathcal{I}^\urcorner=\mathcal{I}_1^\urcorner+\mathcal{I}_2^\urcorner,\qquad \mathcal{I}_1^\urcorner\cap\mathcal{I}_2^\urcorner=\mathcal{K}(\ell^2(\NN^2_\urcorner)).$$
 \end{lemma}
 \begin{proof}
Observe that the projections onto the translated faces $(0,1)\cdot F_1$ and $(1,0)\cdot F_2$ are given by $V_yP_{F_1^\urcorner}V_y^*$ and $V_xP_{F_2^\urcorner}V_x^*$, so that they lie in $\mathcal{I}_1^\urcorner$ and $\mathcal{I}_2^\urcorner$ respectively; the product of these translated face-projections is the rank-1 projection $P_{(1,1)}\in \mathcal{I}_1^\urcorner\cap\mathcal{I}_2^\urcorner$. This shows that\footnote{In particular, the third commutator $[V_y^*,V_x]=|(1,0)\rangle\langle(0,1)|$ in Eq.\ \eqref{commutators} is already present in $\mathcal{I}_1^\urcorner\cap\mathcal{I}_2^\urcorner$ and does not generate anything outside $\mathcal{I}_1^\urcorner+\mathcal{I}_2^\urcorner$.} $\mathcal{K}(\ell^2(\NN^2_\urcorner))\subset\mathcal{I}_1^\urcorner\cap\mathcal{I}_2^\urcorner$. For the reverse inclusion $\mathcal{K}(\ell^2(\NN^2_\urcorner))\supset\mathcal{I}_1^\urcorner\cap\mathcal{I}_2^\urcorner$, let $C\in\mathcal{I}_1^\urcorner\cap\mathcal{I}_2^\urcorner$, then it can be norm-approximated by a $C'\in\mathcal{I}_1^\urcorner\cap\mathcal{I}_2^\urcorner$ that is expressible as a finite linear combination of terms of the form $Ap_{F_1^\urcorner}A'$ or of the form $Bp_{F_2^\urcorner}B'$, with $A,A',B,B'$ some finite product of the $V_x,V_x^*,V_y,V_y^*$. Such an operator $C'$ is zero on all $|(m,n)\rangle$ except for finitely many $m$ and for finitely many $n$, so $C'$ is a finite-rank approximation of $C$.
 \end{proof}

Up to commutator-terms involving $[V_y^*,V_x]$ or $[V_x^*,V_y]$ (which are compact operators), an element of $\mathcal{I}_1^\urcorner$ is approximated by a linear sum of terms like $$\underbrace{V_y^aP_{F_1^\urcorner}(V_y^*)^b}_{{\rm finite-rank\, in}\, y}\underbrace{(V_x^*)^{c_1}V_x^{d_1}\ldots (V_x^*)^{c_k}V_x^{d_k}}_{(V_x^*)^{c}V_x^{d}\,{\rm up\, to\, }P_{F_2^\urcorner}\,{\rm terms}}\quad\overset{{\rm mod}\,\mathcal{I}_1^\urcorner\cap\mathcal{I}_2^\urcorner}{\mapsto} \quad U_x^{d-c}\otimes |a\rangle\langle b|,$$
so after modding out by $\mathcal{K}(\ell^2(\NN^2_\urcorner))=\mathcal{I}_1^\urcorner\cap\mathcal{I}_2^\urcorner$, there is an isomorphism
$$\mathcal{I}_1^\urcorner/\mathcal{K}(\ell^2(\NN^2_\urcorner))\overset{\cong}{\rightarrow} C^*_r(\ZZ)\otimes\mathcal{K}(\ell^2(\NN)).$$
Similarly for $\mathcal{I}_2^\urcorner$ with the roles of $x$ and $y$ switched. Thus we have
\begin{eqnarray}
\mathcal{I}^\urcorner/\mathcal{K}(\ell^2(\NN^2_\urcorner))=(\mathcal{I}_1^\urcorner+\mathcal{I}_2^\urcorner)/(\mathcal{I}_1^\urcorner\cap\mathcal{I}_2^\urcorner)&\cong& \mathcal{I}_1^\urcorner/(\mathcal{I}_1^\urcorner\cap\mathcal{I}_2^\urcorner)\oplus \mathcal{I}_2^\urcorner/(\mathcal{I}_1^\urcorner\cap\mathcal{I}_2^\urcorner) \nonumber\\
&\cong & (C^*_r(\ZZ)\otimes \mathcal{K})\oplus(\mathcal{K}\otimes C^*_r(\ZZ)),\label{eqn:ideal.modK}
\end{eqnarray}
with quotient map $q$ explicitly given by
$$\mathcal{I}_1^\urcorner\ni V_y^aP_{F_1^\urcorner}(V_y^*)^b(V_x^*)^cV_x^d + \mathcal{K} \mapsto (U_x^{d-c}\otimes |a\rangle\langle b|,0), $$
\begin{equation}
\mathcal{I}_2^\urcorner\ni V_x^aP_{F_2^\urcorner}(V_x^*)^b(V_y^*)^cV_y^d + \mathcal{K} \mapsto (0,|a\rangle\langle b| \otimes U_y^{d-c}).\label{q.formula}
\end{equation}
\noindent
We have shown:
\begin{lemma}\label{lem:bumpy.SES}
There is a short exact sequence
\begin{equation} 0\rightarrow \mathcal{K}\rightarrow \mathcal{I}^\urcorner\overset{q}{\rightarrow}\mathcal{I}^\urcorner/\mathcal{K}\cong (C^*_r(\ZZ)\otimes\mathcal{K})\oplus(\mathcal{K}\otimes C^*_r(\ZZ))\rightarrow 0.\label{subideal.SES}
\end{equation}
\end{lemma}

The $K$-theory LES for Eq.\ \eqref{subideal.SES} is
\begin{equation}\label{subideal.LES}
\xymatrix{
{\overbrace{K_0(\mathcal{K})}^{\ZZ}} \ar[r]^{0} & {\overbrace{K_0(\mathcal{I}^\urcorner)}^{\ZZ\oplus\ZZ}} \ar[r]^{q_*} & {\overbrace{K_0(\mathcal{I}^\urcorner/\mathcal{K})}^{\ZZ\oplus\ZZ}} \ar[d]^{{\rm Exp}} \\
{\underbrace{K_1(\mathcal{I}^\urcorner/\mathcal{K})}_{\ZZ\oplus\ZZ}} \ar[u]^{{\rm Ind}={\rm -sum}} & {\underbrace{K_1(\mathcal{I}^\urcorner)}_{\ZZ}} \ar[l]_{q_*} & {\underbrace{K_1(\mathcal{K})}_{0}}\ar[l]_{0}
}
\end{equation}
Here, the index map can be deduced by its action on the representative generators of $K_1(\mathcal{I}^\urcorner/\mathcal{K})\cong K_1(C^*_r(\ZZ)\otimes\mathcal{K})\oplus K_1(\mathcal{K}\otimes C^*_r(\ZZ))$. These are $U_x\otimes |0\rangle\langle 0|$ for the first direct sum factor and $ |0\rangle\langle 0|\otimes U_y$ for the second factor. From Eq.\ \eqref{q.formula}, a lift of $U_x\otimes |0\rangle\langle 0|$ in $\mathcal{I}^\urcorner$ is $V_xP_{F_1^\urcorner}$, which is just the unilateral shift on $\ell^2(F_1)\cong\ell^2(\NN)$ which has Fredholm index $-1$. Thus ${\rm Ind}([U_x])=-[P_{\{(0,1)\}}]=-1$; similarly,  ${\rm Ind}([U_y])=-[P_{\{(1,0)\}}]=-1$. From Eq.\ \eqref{q.formula} applied to the LES, we also see that 
\begin{equation}
K_0(\mathcal{I}^\urcorner)\cong \ZZ [P_{F_1^\urcorner}]\oplus \ZZ[P_{F_2^\urcorner}],\label{eqn:staircase.K0.ideal.generators}.
\end{equation}
and that $K_1(\mathcal{I}^\urcorner)$ is identified inside $K_1(\mathcal{I}^\urcorner/\mathcal{K})\cong (\mathcal{K}\otimes C^*_r(\ZZ))\oplus(\mathcal{K}\otimes C^*_r(\ZZ))$ as ${\rm ker}({\rm Ind})\cong\ZZ$. The latter kernel has representative generator $[(U_x\otimes |0\rangle\langle 0|\,,\,|0\rangle\langle 0|\otimes U_y^*)]$, so that 
its lift 
\begin{equation}
w_\urcorner:=V_xP_{F_1^\urcorner}+V_y^*(P_{F_2^\urcorner}+P_{\{(1,1)\}})\in\mathcal{I}^\urcorner\label{eqn:staircase.cornering}
\end{equation}
represents the generator of $K_1(\mathcal{I}^\urcorner)\cong\ZZ$. This $w_\urcorner$ is illustrated in Fig.\ \ref{fig:staircase}, along with an alternative ``smoother'' representative
\begin{equation}
w'_\urcorner=V_yV_x(P_{F_1^\urcorner}-P_{\{(0,1)\}})V_y^*+V_xV_y^*P_{F_2^\urcorner}V_x^*.\label{alternative.cornering}
\end{equation}

We can now analyse the desired LES for the staircase boundary condition SES Eq.\ \eqref{staircase.SES}:
\begin{equation*}%\label{incomplete.staircase.LES}
\xymatrix{
{\overbrace{K_0(\mathcal{I}^\urcorner)}^{\ZZ [P_{F_1^\urcorner}]\oplus \ZZ[P_{F_2^\urcorner}]}} \ar[r]^{0} & {\overbrace{K_0(C^*_r(\NN^2_\urcorner))}^{?\oplus\ZZ[1]}} \ar[r]^{\pi_*} & {\overbrace{K_0(C^*_r(\ZZ^2))}^{\ZZ[1]\oplus\ZZ[\mathfrak{b}]}} \ar[d]^{{\rm Exp}} \\
{\underbrace{K_1(C^*_r(\ZZ^2))}_{\ZZ[U_x]\oplus\ZZ[U_y]}} \ar[u]^{{\rm Ind}}_{\cong} & {\underbrace{K_1(C^*_r(\NN^2_\urcorner))}_{?}} \ar[l]_{\pi_*=0} & {\underbrace{K_1(\mathcal{I}^\urcorner)}_{\ZZ[w_\urcorner]}}\ar[l]_{?}
}
\end{equation*}
The index map may be computed to be an isomorphism by considering generators (e.g.\ ${\rm Ind}([U_x])=-[P_{F_2^\urcorner}]$). To complete the diagram, we need to argue that $K_1(C^*_r(\NN^2_\urcorner))=0$. A priori, $K_1(C^*_r(\NN^2_\urcorner))$ should be given by the image of $K_1(\mathcal{I}^\urcorner)\cong\ZZ[w_\urcorner]$ under inclusion, which we now compute. We use the alternative generator $w'_\urcorner$ from Eq.\ \eqref{alternative.cornering}, which resembles the generator $w=\hat{U}_xP_{F_1}+\hat{U}_y^*P_{F_2}$ of $K_1(\mathcal{I}')\cong\ZZ$ encountered in Section \ref{sec:standard.quarter.plane}. There, we saw that $[w]\in K_1(C^*_r(\NN^2))=0$ necessarily trivialises\footnote{Note the convention that inside the unital algebra $C^*_r(\NN^2)$, $w$ is regarded as the identity operator everywhere away from the boundary.}. This means that inside $C^*_r(\NN^2)$, $w$ may be deformed through unitaries into the identity operator (after passing to unitisations and perhaps matrix algebras). The same deformation now performed on the shifted quarter-plane with corner at $(1,1)$ will have the effect of turning $w'_\urcorner$ into the identity operator in $C^*_r(\NN^2_\urcorner)$. The upshot is that $[w_\urcorner]=[w'_\urcorner]$ again trivialises when mapped into $K_1(C^*_r(\NN^2_\urcorner))$. Thus $K_1(C^*_r(\NN^2_\urcorner))=0$, and the above LES is completed as:

\begin{theorem}\label{thm:staircase.LES}
There is a long exact sequence,
\begin{equation}\label{staircase.LES}
\xymatrix{
{\overbrace{K_0(\mathcal{I}^\urcorner)}^{\ZZ [P_{F_1^\urcorner}]\oplus \ZZ[P_{F_2^\urcorner}]}} \ar[r]^{0} & {\overbrace{K_0(C^*_r(\NN^2_\urcorner))}^{\ZZ[1]}} \ar[r]^{\pi_*} & {\overbrace{K_0(C^*_r(\ZZ^2))}^{\ZZ[1]\oplus\ZZ[\mathfrak{b}]}} \ar[d]^{{\rm Exp}}_{[\mathfrak{b}]\mapsto[w_\urcorner]} \\
{\underbrace{K_1(C^*_r(\ZZ^2))}_{\ZZ[U_x]\oplus\ZZ[U_y]}} \ar[u]^{{\rm Ind}}_{\cong} & {\underbrace{K_1(C^*_r(\NN^2_\urcorner))}_{0}} \ar[l]_{\pi_*=0} & {\underbrace{K_1(\mathcal{I}^\urcorner)}_{\ZZ[w_\urcorner]}}\ar[l]
}
\end{equation}
\end{theorem}

\subsubsection{General imperfect quarter-plane boundaries}\label{sec:standard.imperfect}

For general quarter-plane boundary conditions, say with multiple steps of varying length, bumps, angles etc., essentially the same arguments can be used. Let us sketch how this works.

Consider the standard quarter-plane, with imperfections confined within a square $[-R,R]\times[-R,R]$, as illustrated in Fig.\ \ref{fig:bumpy.corner}. This models, near its corner, a material with imperfect boundaries. Let $\NN^2_\urcorner$ now denote the subset of lattice points included in the material. As before, we can construct $C^*_r(\NN^2_\urcorner)\subset \mathcal{B}(\ell^2(\NN^2_\urcorner))$, generated by the truncations $V_\gamma$ of the $U_\gamma, \gamma\in\NN^2$ to $\ell^2(\NN^2_\urcorner)$. A little thought shows that this $C^*_r(\NN^2_\urcorner)$ does contain all the truncated $V_\gamma, \gamma\in \ZZ^2$ --- write $\gamma=\gamma_-^{-1}\gamma_+$ for some $\gamma_\pm\in S$ with $\gamma_+$ ``large enough'' to ensure $U_{\gamma_+}\circ\iota=\iota\circ p \circ U_{\gamma_+}\circ\iota\equiv V_{\gamma_+}$, then we can write $V_\gamma=V_{\gamma_-}^*V_{\gamma_+}$.

Up to some finite-rank projection, the commutator $[V_x^*,V_x]$ is again a projection onto a vertical face $F_2$, and similarly for $[V_y^*,V_y]$ (see Fig. \ref{fig:bumpy.corner}). Sufficiently shifted versions of these two commutators multiply into a rank-one projection, since the shifted faces eventually intersect at a single point (see Fig.\ \ref{fig:bumpy.corner}). So the compact operators are in the the commutator ideal $\mathcal{I}^\urcorner\subset C^*_r(\NN^2_\urcorner)$. For $\gamma=(m,n)\in\NN^2$, the difference $V_{(m,n)}-V_x^mV_y^n$ is at most finite-rank. This means that $C^*_r(\NN^2_\urcorner)$ is already generated from the basic partial isometries $V_x$ and $V_y$. As before we can define $\mathcal{I}_1^\urcorner$ and $\mathcal{I}_2^\urcorner$ to be the ideals in $C^*_r(\NN^2_\urcorner)$ generated respectively by $P_{F_1},P_{F_2}$, and we have $\mathcal{K}\subset\mathcal{I}_1^\urcorner\cap\mathcal{I}_2^\urcorner$ (the other basic commutator $[V_y^*,V_x]$ is finite-rank, so does not generate anything extra).

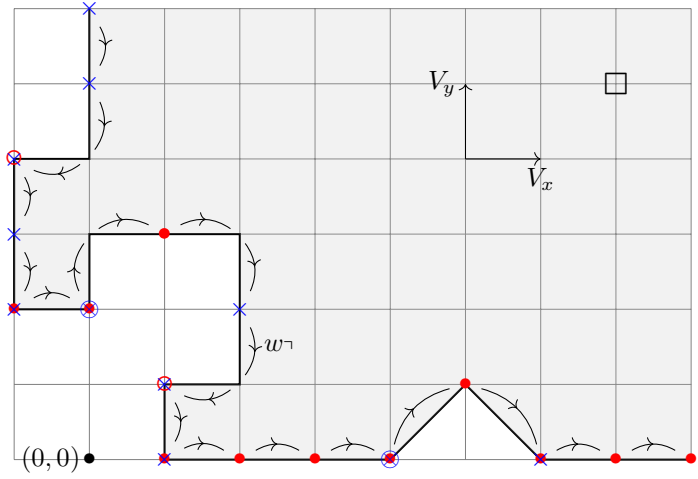
\begin{figure}
\begin{tikzpicture}
\draw[help lines] (0,0) grid (9,6);
\draw[thick] (1,6) -- (1,4) -- (0,4) -- (0,2) -- (1,2) -- (1,3) -- (3,3) -- (3,1) -- (2,1) -- (2,0) -- (5,0) -- (6,1) -- (7,0) -- (9,0);
\fill[lightgray, opacity=0.2] (1,6) -- (1,4) -- (0,4) -- (0,2) -- (1,2) -- (1,3) -- (3,3) -- (3,1) -- (2,1) -- (2,0) -- (5,0) -- (6,1) -- (7,0) -- (9,0) -- (9,6) -- (0,6);

\node[red] at (0,2) {$\bullet$};
\node[red] at (1,2) {$\bullet$};
\node[red] (i) at (2,3) {$\bullet$};
\node (j) at (3,3) {};
\node[red] at (2,3) {$\bullet$};
\node[red] at (2,0) {$\bullet$};
\node[red] (o) at (3,0) {$\bullet$};
\node[red] (p) at (4,0) {$\bullet$};
\node[red] (q) at (5,0) {$\bullet$};
\node[red] (r) at (6,1) {$\bullet$};
\node[red] (s) at (7,0) {$\bullet$};
\node[red] (t) at (8,0) {$\bullet$};
\node[red] (u) at (9,0) {$\bullet$};
\node[red] at (2,1) {{\Large $\circ$}};
\node[red] at (0,4) {{\Large $\circ$}};

\node[blue] (b) at (1,5) {$\times$};
\node[blue] (a) at (1,6) {$\times$};
\node[blue] (c) at (1,4) {};
\node[blue] (d) at (0,4) {$\times$};
\node[blue] (e) at (0,3) {$\times$};
\node[blue] (f) at (0,2) {$\times$};
\node[blue] (k) at (3,2) {$\times$};
\node (l) at (3,1) {};
\node[blue] (m) at (2,1) {$\times$};
\node[blue] (n) at (2,0) {$\times$};
\node[blue] at (7,0) {$\times$};
\node[blue] (g) at (1,2) {$\otimes$};
\node (h) at (1,3) {};
\node[blue] at (5,0) {$\otimes$};

\node at (8,5) {{\Large $\Box$}};

\path
    (a) edge[bend left, ->-] node [right] {} (b)
    (b) edge[bend left, ->-] node [right] {} (c)
    (c) edge[bend left, ->-] node [right] {} (d)
    (d) edge[bend left, ->-] node [right] {} (e)
    (e) edge[bend left, ->-] node [right] {} (f)
    (f) edge[bend left, ->-] node [right] {} (g)
    (g) edge[bend left, ->-] node [right] {} (h)
    (h) edge[bend left, ->-] node [right] {} (i)
    (i) edge[bend left, ->-] node [right] {} (j)
    (j) edge[bend left, ->-] node [right] {} (k)
    (k) edge[bend left, ->-] node [right] {} (l)
    (l) edge[bend left, ->-] node [right] {} (m)
    (m) edge[bend left, ->-] node [right] {} (n)
    (n) edge[bend left, ->-] node [right] {} (o)
    (o) edge[bend left, ->-] node [right] {} (p)
    (p) edge[bend left, ->-] node [right] {} (q)
    (q) edge[bend left, ->-] node [right] {} (r)
    (r) edge[bend left, ->-] node [right] {} (s)
    (s) edge[bend left, ->-] node [right] {} (t)
    (t) edge[bend left, ->-] node [right] {} (u);

\node[right] at (3.2,1.5) {$w_\urcorner$};

\node (A) at (4,6) {};
\node (B) at (4,5) {};
\node (C) at (4,4) {};
\node (D) at (5,4) {};
\node (E) at (6,4) {};
\node (F) at (7,4) {};
\node (G) at (8,4) {};
\node (H) at (9,4) {};

\draw[->] (6,4) -- (6,5) {};
\draw[->] (6,4) -- (7,4) {};
\node[below] at (7,4) {$V_x$};
\node[left] at (6,5) {$V_y$};

\node at (1,0) {$\bullet$};
\node[left] at (1,0) {$(0,0)$};

\end{tikzpicture}
\caption{Possible imperfection of the boundary near the corner of a quarter-plane. The horizontal and vertical face projections are indicated by the red $\red{\bullet}$ and blue $\blue{\times}$ respectively, and are the respective commutators $[V_x^*,V_x]$, $[V_x^*,V_x]$ up to a projection onto some finite set of points indicated by ${\Large{\red{\circ}}}$ and ${\Large{\blue{\circ}}}$. Sufficiently shifted versions of the two face projections multiply into a rank-1 projection (${\Large{\Box}}$).}\label{fig:bumpy.corner}
\end{figure}

With the above definitions and considerations, we may check that Lemmas \ref{lem:staircase.commutator.ideal}, \ref{lem:commutator.ideal}, \ref{lem:bumpy.SES}, and the computations leading to Theorem \ref{thm:staircase.LES}, continue to hold in the setting of quarter-planes with imperfect boundaries.

\subsection{Rational slope quarter-plane $C^*$-algebra}\label{sec:quarter.plane.semigroup}
Next, consider a material occupying a closed convex cone $C$ in the Euclidean plane which is pointed and nondegenerate (to exclude the half-plane and half-line cases), whose corner is taken to be at the origin. We also call $C$ a \emph{quarter-plane}. The lattice points in the material are labelled by the intersection $S=C\cap \ZZ^2$, which is a subsemigroup of $\ZZ^2$ (with identity). 

We again have the inclusion $\iota:\ell^2(S)\rightarrow\ell^2(\ZZ^2)$ and projection $p=\iota^*:\ell^2(\ZZ^2)\rightarrow\ell^2(S)$, and can define the truncated translation operators $\hat{U}_\gamma=p\circ U_\gamma\circ\iota, \gamma\in \ZZ^2$. Generalising Definition \ref{defn:standard.quarter.plane}, we have

\begin{definition}\label{defn:general.quarter.plane}
Let $S=C\cap\ZZ^2\subset\ZZ^2$ be the subsemigroup associated to a (pointed, nondegenerate) closed convex cone $C$ as above. We define $C^*_r(S)$ to be the unital $C^*$-subalgebra of $\mathcal{B}(\ell^2(S))$ generated by $\hat{U}_\gamma, \gamma\in S$.
\end{definition}
\noindent
Because $S$ generates $\ZZ^2$, any $\gamma\in\ZZ^2$ can be written as $\gamma=\gamma_-^{-1}\gamma_+$ with $\gamma_\pm\in S$, so that $\hat{U}_\gamma=p\circ U_{\gamma_-}^*U_{\gamma_+}\circ\iota =p\circ U_{\gamma_-}^*\circ\iota\circ p\circ U_{\gamma_+}\circ \iota =\hat{U}_{\gamma_-}^*\hat{U}_{\gamma_+}$. Thus we could equivalently define $C^*_r(S)$ to be the $C^*$-algebra generated by $\hat{U}_\gamma, \gamma\in\ZZ^2$.

\vspace{0.5em}

\emph{In this subsection, we restrict to quarter-planes with rational slopes} (for irrational slopes, see Section \ref{sec:irrational.computations}). Formally, this means that the cone $C$ is \emph{integral}, i.e.\ its two extremal rays have rational slopes (possibly $\pm \infty$). Then $S=C\cap\ZZ^2$ is finitely-generated, generates $\ZZ^2$ (i.e.\ $S-S=\ZZ^2$), and saturated (i.e.\ if some positive multiple of $\gamma\in\ZZ^2$ lies in $S$, then already $\gamma\in S$), e.g.\ \cite{Neeb} Lemma II.5. We are in a special case of \cite{Cuntz}, and can compute the $K$-theory of the semigroup $C^*$-algebra $C^*_r(S)$. 

\begin{remark}
The $K$-theory of $C^*_r(S)$ was first studied in \cite{ParkThesis, ParkSchochet}, while recent work of \cite{Cuntz} provides a more direct computation in the rational slope case. However, \cite{Cuntz} is primarily concerned with finitely-generated subsemigroups $S\subset\ZZ^2$ without the saturation assumption. In the rest of this subsection, we distil from \cite{Cuntz} a minimal path to $K_\bullet(C^*_r(S))$ (Theorem \ref{thm:big.LES}), for the (automatically saturated) $S$ that appear in our physical problem.
\end{remark}

The subset of lattice points $F_i$ lying on an extremal ray are given by $n(x_i,y_i), n\in\NN$ for some integers $x_i,y_i$, $i=1,2$. As in \cite{Cuntz}, we call $F_1,F_2$ the two \emph{faces} of $S$, and $a_i=(x_i,y_i)\in\ZZ^2$ the two \emph{asymptotic generators} of $S$, ordering them in such a way that the following integral matrices have positive determinant (so $F_1$ rotates anticlockwise onto $F_2$ by an angle smaller than $\pi$):
\begin{equation}
M=\begin{pmatrix}y_2 & -x_2 \\ -y_1 & x_1\end{pmatrix},\qquad M^\perp=\begin{pmatrix}x_1 & x_2 \\ y_1 & y_2\end{pmatrix}.
\end{equation}
Note that ${\rm det}(M)={\rm det}(M^\perp)$ and $MM^\perp=M^\perp M=({\rm det}\, M)\mathbf{1}$, and equals the number of lattice points in the parallelogram defined by $a_1\wedge a_2$.

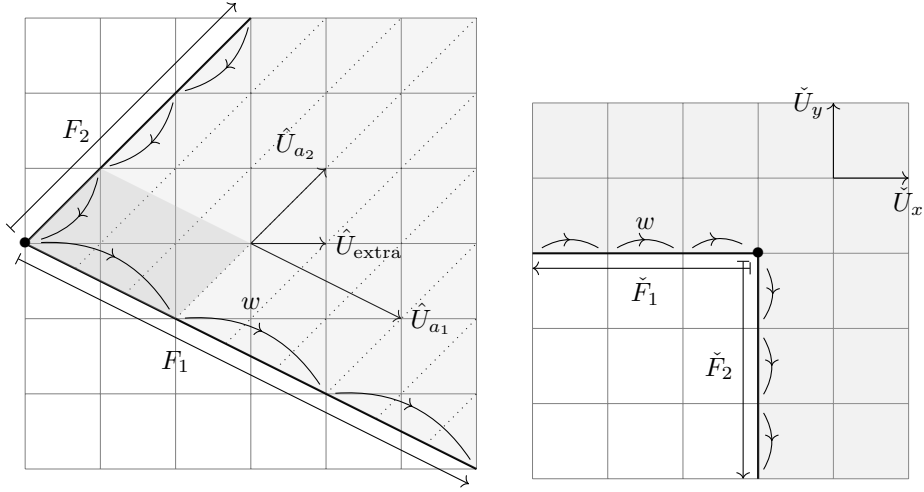
\begin{figure}
\begin{tikzpicture}
\draw[help lines] (0,0) grid (6,6);
\draw[thick] (6,0) -- (4,1) -- (2,2) -- (0,3) -- (1,4) -- (2,5) -- (3,6);

\draw[thin, dotted] (0.67,2.67) -- (4,6);
\draw[thin, dotted] (1.33,2.33) -- (5,6);
\draw[thin, dotted] (2,2) -- (6,6);
\draw[thin, dotted] (2.67,1.67) -- (6,5);
\draw[thin, dotted] (3.33,1.33) -- (6,4);
\draw[thin, dotted] (4,1) -- (6,3);
\draw[thin, dotted] (4.67,0.67) -- (6,2);
\draw[thin, dotted] (5.33,0.33) -- (6,1);

\draw [->] (3,3) -- (5,2);
\draw [->] (3,3) -- (4,4);
\draw [->] (3,3) -- (4,3);

\draw [|->] (-0.2,3.2) -- (2.8,6.2);
\draw [|->] (-0.1,2.8) -- (5.9,-0.2);

\fill[lightgray, opacity=0.15] (0,3) -- (6,0) -- (6,6) -- (3,6) -- (0,3);
\fill[lightgray, opacity=0.3] (0,3) -- (2,2) -- (3,3) -- (1,4) -- (0,3);

\node (a) at (3,6) {};
\node (b) at (2,5) {};
\node (c) at (1,4) {};
\node (d) at (0,3) {$\bullet$};
\node (e) at (2,2) {};
\node (f) at (4,1) {};
\node (g) at (6,0) {};

\path
    (a) edge[bend left, ->-] node [right] {} (b)
    (b) edge[bend left, ->-] node [right] {} (c)
    (c) edge[bend left, ->-] node [right] {} (d)
    (d) edge[bend left, ->-] node [right] {} (e)
    (e) edge[bend left, ->-] node [right] {} (f)
    (f) edge[bend left, ->-] node [right] {} (g);

\node [above] at (3,2) {$w$};

\node [left] at (1,4.5) {$F_2$};
\node [below] at (2,1.7) {$F_1$};

\node[right] at (5,2) {$\hat{U}_{a_1}$};
\node[above left] at (4,4) {$\hat{U}_{a_2}$};
\node[right] at (4,3) {$\hat{U}_{\rm extra}$};

\end{tikzpicture}
\hspace{1em}
\begin{tikzpicture}
\draw[help lines] (-3,-3) grid (2,2);
\draw[thick] (-3,0) -- (0,0);
\draw[thick] (0,0) -- (0,-3);
\fill[lightgray, opacity=0.2] (-3,0) -- (0,0) -- (0,-3) -- (2,-3) -- (2,2) -- (-3,2) -- (-3,0);

\node (a) at (-3,0) {};
\node (b) at (-2,0) {};
\node (c) at (-1,0) {};
\node (d) at (0,0) {$\bullet$};
\node (e) at (0,-1) {};
\node (f) at (0,-2) {};
\node (g) at (0,-3) {};
\path
    (a) edge[bend left, ->-] node [right] {} (b)
    (b) edge[bend left, ->-] node [right] {} (c)
    (c) edge[bend left, ->-] node [right] {} (d)
    (d) edge[bend left, ->-] node [right] {} (e)
    (e) edge[bend left, ->-] node [right] {} (f)
    (f) edge[bend left, ->-] node [right] {} (g);

\draw[->] (1,1) -- (1,2);
\draw[->] (1,1) -- (2,1);
\node[below] at (2,1) {$\check{U}_x$};
\node[left] at (1,2) {$\check{U}_y$};

\draw[|->] (-0.1,-0.2) -- (-3,-0.2);
\draw[|->] (-0.2,-0.1) -- (-0.2,-3);
\node[below] at (-1.5,-0.2) {$\check{F}_1$};
\node[left] at (-0.2,-1.5) {$\check{F}_2$};

\node[above] at (-1.5,0.2) {$w$};

\end{tikzpicture}
\caption{(L) Quarter-plane with rational slope faces. A fundamental domain is indicated by the dark parallelogram. A foliation of $S$ by translates of face $F_2$ is indicated by the dotted lines. (R) A concave quarter-plane.}\label{fig:rational.angle.concave}
\end{figure}

The range projections of the partial isometries obtained as all possible products of $\hat{U}_\gamma, \hat{U}_\gamma^*, \gamma\in S$
commute among each other, generating a commutative $C^*$-subalgebra $D\subset C^*_r(S)$. For each subset $Y\subset S$, denote by $P_Y$ the projection from $\ell^2(S)$ onto $\ell^2(Y)$. It is shown in \cite{Cuntz}, Lemma 7.3.2, that $P_{F_1},P_{F_2}\in  D\subset C^*_r(S)$ for the two face projections; specifically, for $i=1,2$ one may construct $\gamma_i\in\ZZ^2$ such that $P_{F_i}=1-\hat{U}_{\gamma_i}\hat{U}_{\gamma_i}^*$. For example, in Fig.\ \ref{fig:rational.angle.concave} we could take $\gamma_1=(-1,1)$ and $\gamma_2=(0,-1)$. Thus we may define the ideals $\mathcal{I}_1, \mathcal{I}_2\subset C^*_r(S)$ generated respectively by $P_{F_1},P_{F_2}$. Remarkably, we have, for every $S$ as Definition \ref{defn:general.quarter.plane},
\begin{lemma}[cf.\ 7.2.10, 7.3.6 of \cite{Cuntz}]\label{quarter.plane.SES}
$\mathcal{I}_1\cap\mathcal{I}_2=\mathcal{K}(\ell^2(S))$, and $C^*_r(S)/(\mathcal{I}_1+\mathcal{I}_2)\cong C^*_r(\ZZ^2)$ via the map $\pi:\hat{U}_\gamma\mapsto U_\gamma$. Thus, writing $\mathcal{I}:=\mathcal{I}_1+\mathcal{I}_2$, there is a short exact sequence of $C^*$-algebras,
\begin{equation}
0\rightarrow \mathcal{I}\rightarrow C^*_r(S)\overset{\pi}{\rightarrow} C^*_r(\ZZ^2)\rightarrow 0.  \label{general.quarter.plane.SES}
\end{equation}
\end{lemma}
\begin{proof}
$P_{F_1}P_{F_2}\in \mathcal{I}_1\cap\mathcal{I}_2$ is the rank-1 projection $P_{\{0,0\}}$ onto the origin, so $\mathcal{I}_1\cap\mathcal{I}_2=\mathcal{K}(\ell^2(S))$ follows from the same argument as in the proof of Lemma \ref{lem:commutator.ideal}.
For each $\gamma\in S$, the complement $S\setminus (\gamma\cdot S)$ is a union of finitely many translates of the faces $F_1, F_2$, so the projection onto this complement is in $\mathcal{I}=\mathcal{I}_1+\mathcal{I}_2$. Thus $\hat{U}_\gamma$ is unitary modulo $\mathcal{I}$.
\end{proof}

Due to $[\hat{U}_{a_1}^*,\hat{U}_{a_2}]=0$, the extension Eq. \eqref{general.quarter.plane.SES} is related to the standard extension $0\rightarrow\mathcal{I}'\rightarrow C^*_r(\NN)^2\rightarrow C^*_r(\ZZ^2)\rightarrow 0$ (Eq. \eqref{standard.quarter.plane.SES}):
\begin{lemma}[7.3.8 of \cite{Cuntz}]\label{lem:subalgebra.isomorphic.to.standard}
The $C^*$-subalgebra of $C^*_r(S)$ generated by $\hat{U}_{a_1}$ and $\hat{U}_{a_2}$ is isomorphic to $C^*_r(\NN^2)$. Explicitly, there is a (injective) morphism $\kappa: C^*_r(\NN^2)\rightarrow C^*_r(S)$ defined on generators by 
$$\hat{U}_x\mapsto\hat{U}_{a_1},\qquad \hat{U}_y\mapsto\hat{U}_{a_2}.$$
\end{lemma}
\noindent 
Thus there is the following commutative diagram associated to $\kappa$,
\begin{equation}\label{SES.relation}
\xymatrix{
0\ar[r] & \mathcal{I}\ar[r] & C^*_r(S) \ar[r]^{\pi} & C^*_r(\ZZ^2)\ar[r] & 0 \\
0\ar[r] & \mathcal{I}'\ar[r]\ar[u]^{\kappa |} & C^*_r(\NN^2) \ar[r]^{\pi}\ar[u]^{\kappa}& C^*_r(\ZZ^2)\ar[r]\ar[u]^{\check{\kappa}} & 0 
}
\end{equation}
where $\kappa |$ is the restriction of $\kappa$ to $\mathcal{I}'$, and $\check{\kappa}$ is the induced map from the group homomorphism $M^\perp:\ZZ^2\rightarrow\ZZ^2$.

The $K$-theory LES for the upper short exact sequence and for the lower one can be combined into a single big diagram, because the connecting homomorphisms in $K$-theory are \emph{natural} with respect to morphisms of short exact sequences. This observation will allow the desired six-term LES of Eq.\ \eqref{general.quarter.plane.SES} (the upper sequence) to be computed from the known standard LES (Eq.\ \eqref{standard.quarter.plane.LES}) for the lower sequence, together with knowledge of the functorially induced morphisms $(\kappa |)_*,\check{\kappa}_*$.

\begin{lemma}[7.3.9(1) of \cite{Cuntz}]\label{lem:ideal.k-theory}
For any subsemigroup $S\subset\ZZ^2$ as in Definition \ref{defn:general.quarter.plane}, we have $K_0(\mathcal{I})\cong K_0(\mathcal{I}')$ (thus they are each isomorphic to $\ZZ^2$ due to Eq.\ \eqref{standard.quarter.plane.LES}). Also, $K_1(\mathcal{I})\cong\ZZ$ with generator represented by $w=\hat{U}_{a_2}P_{F_1}+\hat{U}_{a_1}^*P_{F_2}$.
\end{lemma}
\noindent
\begin{remark}\label{rem:edge.travelling.op}
The representative generator $w$ is just the anticlockwise translation operator on the boundary Hilbert space $\ell^2(F_1\cup F_2)$, see Fig.\ \ref{fig:rational.angle.concave}. We call this the ``edge-travelling operator''. Quite generally, such an edge-travelling operator represents the $K_1$-theory generator of the coarse index of the Dirac operator on a noncompact 1D manifold, see \S 5.1 of \cite{Ludewig-Thiang-cobordism}.
\end{remark} 

\begin{lemma}[7.3.9 (2)-(4) of \cite{Cuntz}]\label{lem:functorial.maps}
Let $\kappa: C^*_r(\NN^2)\rightarrow C^*_r(S)$ be as in Lemma \ref{lem:subalgebra.isomorphic.to.standard}. Then
\begin{enumerate}
\item $\check{\kappa}_*:\ZZ^2\cong K_0(C^*_r(\ZZ^2))\rightarrow K_0(C^*_r(\ZZ^2))\cong\ZZ^2$ takes 
$$[1]\mapsto[1],\qquad [\mathfrak{b}]\mapsto ({\rm det}\,M)[\mathfrak{b}].$$
\item The index map for Eq.\ \eqref{general.quarter.plane.SES},
$${\rm Ind}: \ZZ^2\cong K_1(C^*_r(\ZZ^2))\rightarrow K_0(\mathcal{I})\cong\ZZ^2, $$
is given by multiplication by $M$.
\item The induced map $(\kappa |)_*:\ZZ\cong K_1(\mathcal{I}')\rightarrow K_1(\mathcal{I})\cong\ZZ$ is given by multiplication by ${\rm det}\,M$.
\end{enumerate}
\end{lemma}
\noindent
%Although not detailed in \cite{Cuntz}, 
It is easy to check (e.g.\ Eq.\ 7.3 of \cite{Cuntz}) that $(\kappa |)_*:\ZZ^2\cong K_0(\mathcal{I}')\rightarrow K_0(\mathcal{I})\cong\ZZ^2$ is given by multiplication by $({\rm det}\,M){\bf 1}$, and that $\check{\kappa}_*:\ZZ^2\cong K_1(C^*_r(\ZZ^2))\rightarrow K_1(C^*_r(\ZZ^2))\cong\ZZ^2$ is given by multiplication by $M^\perp$. 

In conclusion, we have deduced:
\begin{theorem}[cf.\ 7.3.11 of \cite{Cuntz}]\label{thm:big.LES}
For any subsemigroup $S\subset\ZZ^2$ in the setup of Definition \ref{defn:general.quarter.plane}, the following diagram commutes:
\begin{equation}
\xymatrix{
{\overbrace{K_0(\mathcal{I})}^{\ZZ^2}}\ar[rr] & & {\overbrace{K_0(C^*_r(S))}^{{\rm coker}\,M\oplus\ZZ[1]}} \ar[rr]^{\pi_*} & & {\overbrace{K_0(C^*_r(\ZZ^2))}^{\ZZ[1]\oplus\ZZ[\mathfrak{b}]}} \ar[ddd]^{\rm Exp} \\ 
 & {\overbrace{K_0(\mathcal{I}')}^{\ZZ^2}} \ar[r]^{0}\ar[lu]_{({\rm det}\,M){\bf 1}}^{(\kappa |)_*} & {\underbrace{K_0(C^*_r(\NN^2))}_{\ZZ[1]}} \ar[r]^{\pi_*}\ar[u]_{\kappa^*} & {\overbrace{K_0(C^*_r(\ZZ^2))}^{\ZZ[1]\oplus\ZZ[\mathfrak{b}]}} \ar[d]^{{\rm Exp}}\ar[ur]^{{\rm diag}(1,{\rm det}\,M)}_{\check{\kappa}_*} & \\
 & {\underbrace{K_1(C^*_r(\ZZ^2))}_{\ZZ^2}} \ar[u]^{{\rm Ind}={\bf 1}}\ar[ld]^{M^\perp}_{\check{\kappa}_*} & {\overbrace{K_1(C^*_r(\NN^2))}^{0}} \ar[l]_{0}\ar[d]_{\kappa_*} & {\underbrace{K_1(\mathcal{I}')}_{\ZZ[w]}}\ar[l]_{0}\ar[dr]^{(\kappa |)_*}_{{\rm det}\,M} & \\
{\underbrace{K_1(C^*_r(\ZZ^2))}_{\ZZ^2}}\ar[uuu]^{{\rm Ind}=M} & & {\underbrace{K_1(C^*_r(S))}_{0}}\ar[ll]^{0} & & {\underbrace{K_1(\mathcal{I})}_{\ZZ[w]}}\ar[ll]
}
\end{equation}
The exponential map ${\rm Exp}:K_0(C^*_r(\ZZ^2))\rightarrow K_1(\mathcal{I})$ for the sequence Eq.\ \eqref{general.quarter.plane.SES} maps the Bott generator $[\mathfrak{b}]$ to the anticlockwise winding generator $[w]$ of Lemma \ref{lem:ideal.k-theory}.
\end{theorem}

While we are interested in the exponential map rather than the index map, let us mention for completeness that the index map in Lemma \ref{lem:functorial.maps} should be understood according to the following conventions of \cite{Cuntz}. Generalising Eq.\ \eqref{eqn:ideal.modK}, there is a quotient map (cf.\ Lemma 7.3.6 of \cite{Cuntz}),
\begin{eqnarray}
\mathcal{I}/(\mathcal{I}_1\cap\mathcal{I}_2)&=&\mathcal{I}/\mathcal{K}(\ell^2(S)) \cong \mathcal{I}_1/\mathcal{K}\oplus\mathcal{I}_2/\mathcal{K} \nonumber\\
&\cong& \left(\mathcal{K}(\ell^2(S/F_1))\otimes C^*_r(\ZZ)\right) \oplus \left(C^*_r(\ZZ)\otimes \mathcal{K}(\ell^2(S/F_2))\right), \label{eqn:ideal.modK.rational}
\end{eqnarray}
and then we can choose $P_{F_2}, P_{F_1}$ to represent a basis for $K_0(\mathcal{I})\cong\ZZ^2$. Then $K_0(\mathcal{I})\cong\ZZ[P_{F_2}]\oplus\ZZ[P_{F_1}]$ and $K_1(C^*_r(\ZZ^2))\cong\ZZ[U_x]\oplus\ZZ[U_y]$ with a sign factor included in the index map. We should think of $S/F_i$ as a ``transversal label'' for the foliation of $S$ by translates of $F_i$, see Fig.\ \ref{fig:rational.angle.concave}.

\subsubsection{Rational slope quarter-plane with imperfect boundary}
As in Section \ref{sec:standard.imperfect}, we can introduce imperfections to boundary of a rational slope quarter-plane, confined to a region of finite distance from the origin. Then we define the subset $S^\urcorner$ of lattice points contained in this modified quarter-plane, the truncated translations $V_\gamma, \gamma\in \ZZ^2$, and the algebra $C^*_r(S^\urcorner)$ generated by $V_\gamma, \gamma\in S$. The main complication arising from $S\neq\NN^2$ is that $S$ generally requires some (finite number of) extra generators other than the two asymptotic ones, see Fig.\ \ref{fig:rational.angle.concave}, so the arguments in Section \ref{sec:standard.imperfect} need some modifications, which we sketch here.

We still have the face projections $P_{F_1^\urcorner}, P_{F_2^\urcorner}$ generating ideals $\mathcal{I}_1^\urcorner, \mathcal{I}_2^\urcorner$ in $C^*_r(S^\urcorner)$ as before (although the faces have imperfections near the origin). Let us observe that translate $\gamma_0\cdot S^\urcorner$ is a subset of $S$ for some appropriate ``large'' $\gamma_0$, and that the complement $S^\urcorner\setminus (\gamma_0\cdot S^\urcorner)$ is a disjoint union of translates of $F_1^\urcorner, F_2^\urcorner$ up to some finite set of points. 
The projection onto this complement is thus in $\mathcal{I}_1^\urcorner+\mathcal{I}_2^\urcorner$ so that in the quotient $C^*_r(S^\urcorner)/(\mathcal{I}_1^\urcorner+\mathcal{I}_2^\urcorner)$, each $V_\gamma, \gamma\in S$ becomes unitary. So writing $\mathcal{I}^\urcorner=\mathcal{I}_1^\urcorner+\mathcal{I}_2^\urcorner$, there is again an exact sequence
$$0\rightarrow \mathcal{I}^\urcorner\rightarrow C^*_r(S^\urcorner)\xrightarrow{\pi} C^*_r(\ZZ^2)\rightarrow 0,$$
the ``bumpy version'' of Eq. \eqref{general.quarter.plane.SES}.
As in Section \ref{sec:standard.imperfect}, we conclude that the LES for the above sequence has the same essential property as that for Eq. \eqref{general.quarter.plane.SES}. Namely, namely the exponential map takes $[\mathfrak{b}]\mapsto [w_\urcorner]$ where $w_\urcorner$ is the unitary which translates anticlockwise along the bumpy boundary.

\section{Topological cornering states}\label{sec:topological.cornering.states}
\subsection{Cyclic 1-cocycles for boundary algebra}\label{sec:cyclic.cocycle}
Recall the short exact sequence, Eq.\ \eqref{general.quarter.plane.SES}
$$
0\rightarrow \mathcal{I}\rightarrow C^*_r(S)\overset{\pi}{\rightarrow} C^*_r(\ZZ^2)\rightarrow 0,
$$
for the quarter-plane (with no imperfections) contained between the faces $F_1, F_2$ parallel to $a_1,a_2$ respectively. We will construct a cyclic 1-cocycle \cite{Connes} which pairs with $K_1(\mathcal{I})\cong\ZZ[w]$ to give a quantised boundary current. The ``bumpy'' case requires only minor modifications.

First we define a trace $\tau_{F_1}$ on $\mathcal{I}$ as follows. Recall the quotient map 
$$\mathcal{I}/\mathcal{K}(\ell^2(S))\cong (\mathcal{K}(\ell^2(S/F_1))\otimes C^*_r(\ZZ))\,\oplus (C^*_r(\ZZ)\otimes \mathcal{K}(\ell^2(S/F_2))) $$
from Eq.\ \eqref{eqn:ideal.modK.rational} (the bumpy case is similar, see Eq.\ \ref{eqn:ideal.modK}). There is a canonical trace on $C^*_r(\ZZ)$ (extracting the coefficient of the identity element) and on $\mathcal{K}(\ell^2(S/F_i))$, thus there is a trace $\tau$ on each summand of the RHS. We choose the trace $\frac{1}{|a_1|}\tau\oplus 0$ on $\mathcal{I}/\mathcal{K}(\ell^2(S))$, and pull this back to a trace $\tau_{F_1}$ on $\mathcal{I}$. Similarly, we can define $\tau_{F_2}$ by pulling back $0\oplus\frac{1}{|a_2|}\tau$. With these definitions,  
\begin{equation}
\tau_{F_i}(P_{F_j})=\frac{1}{|a_i|}\delta_{ij}, \qquad \tau_{F_i}({\rm finite\,\, rank})=0, \quad i,j=1,2.\label{eqn:trace.values}
\end{equation} 
We should think of $\tau_{F_i}$ as a trace-per-unit-length along the face $F_i$.

For $i=1,2$, define an unbounded derivation $\partial_i$ on (the usual maximal dense subalgebra $\mathscr{I}$ of) $\mathcal{I}\subset \mathcal{B}(\ell^2(S))$, given by the commutator
$$\partial_i(\,\cdot\,):=[Q_i,\,\cdot\,],\qquad Q_i=\frac{x_iX+y_iY}{|a_i|},$$
representing momentum along the face $F_i$, up to a $\sqrt{-1}$ factor. One easily verifies that for $w=\hat{U}_{a_2}P_{F_1}+\hat{U}_{a_1}^*P_{F_2}$ we have 
\begin{eqnarray}
\partial_i w&=&[Q_i,\hat{U}_{a_1}P_{F_1}]+[Q_i,\hat{U}_{a_2}^*P_{F_2}]\nonumber \\
&=& \frac{1}{|a_i|}\left((a_i\cdot a_1)\hat{U}_{a_1}P_{F_1}-(a_i\cdot a_2)\hat{U}_{a_2}^*P_{F_2}\right)\label{eqn:derivative.anticlockwise.unitary}
\end{eqnarray}
Since $\tau_{F_i}\circ\partial_i=0$, we can define cyclic 1-cocycles (extended to the unitisation and matrix algebras)
$$\xi_i(a,a')=\tau_{F_i}(a\partial_i a'),\qquad a,a'\in\mathscr{I}^+.$$
The element $w\in\mathscr{A}$ is actually the unitary $W=w+1-P_{F_1\cup F_2}$ when regarded in $\mathscr{A}^+$ for $K$-theory computations. 
The pairing of $\xi_i$ with $[w]\in K_1(\mathcal{I})\cong\ZZ$ is
\begin{eqnarray}
\langle[\xi_i],[w]\rangle&:=&\xi_i(W^*-1,W-1)\nonumber\\
&=&\tau_{F_i}((w^*-P_{F_1\cup F_2})\partial_i (w-P_{F_1\cup F_2}))\nonumber\\
&=& \frac{1}{|a_i|}\tau_{F_i}((w^*-P_{F_1\cup F_2})((a_i\cdot a_1)\hat{U}_{a_1}P_{F_1}-(a_i\cdot a_2)\hat{U}_{a_2}^*P_{F_2}))\nonumber\\
&=& \frac{1}{|a_i|}\tau_{F_i}(w^*((a_i\cdot a_1)\hat{U}_{a_1}P_{F_1}-(a_i\cdot a_2)\hat{U}_{a_2}^*P_{F_2}))\nonumber \\
&=& \frac{1}{|a_i|}\tau_{F_i}((a_i\cdot a_1)\hat{U}_{a_1}^*\hat{U}_{a_1}P_{F_1}-(a_i\cdot a_2)\hat{U}_{a_2}\hat{U}_{a_2}^*P_{F_2})\nonumber\\
&=& \frac{1}{|a_i|}\tau_{F_i}((a_i\cdot a_1)P_{F_1}-(a_i\cdot a_2)P_{F_2})\nonumber \\
&=& (-1)^{i+1}.\label{eqn:cyclic.pairing.quarter}
\end{eqnarray}
where we have used Eq.\ \eqref{eqn:derivative.anticlockwise.unitary} and Eq. \eqref{eqn:trace.values}, and also $\tau_{F_i}(\hat{U}_{a_j}P_{F_k})=0=\tau_{F_i}(\hat{U}_{a_j}^*P_{F_k}), i,j,k=1,2$. 

\subsection{Cornering states and quantised boundary currents}\label{sec:cornering.boundary.states}
Consider a Chern insulator with Chern class $k\neq 0$, with bulk Hamiltonian $H$. For any quarter-plane with rational slope boundary, the spectral projection $P_-$ onto energies below its spectral gap has non-vanishing exponential map ${\rm Exp}[P_-]=k[w]\in K_1(\mathcal{I})$, due to $[P_-]$ containing $k$-multiples of the Bott projection class $[\mathfrak{b}]$, and our computation of the ${\rm Exp}$ map in \ref{thm:big.LES}). This guarantees that the quarter-plane truncation $\hat{H}$ of $H$ acquires gap-filling spectra, by exactly the same arguments as in Section \ref{sec:boundary.states}. Namely, in terms of the Hamiltonian $H$, we have
$$k[w]={\rm Exp}[P_-]=[{\rm exp}(-2\pi i \varphi(\hat{H}))]$$
where $\hat{H}$ is $H$ truncated to the quarter-plane (plus possibly some extra self-adjoint boundary term from $\mathscr{I}$), and $\varphi$ is some smooth real-valued function which is 1 below the spectral gap of $H$ and $0$ above the gap. Since the pairing computed in Eq.\ \eqref{eqn:cyclic.pairing.quarter} depends only on $K$-theory classes, so
$$\langle[\xi_i],[{\rm exp}(-2\pi i \varphi(\hat{H}))]\rangle=\langle [\xi_i], k[w]\rangle =k\cdot(-1)^{i+1}\in\ZZ.$$
The LHS gives the quantised boundary current in the direction of $a_i$ contributed by edge states localised near face $F_i$, generalising \S 7.1 of \cite{PSB}, and this is quantised to $k$ units (up to a sign) due to the RHS.

This result continues to hold even if imperfections are introduced into the boundary near the corner, in the sense of Section \ref{sec:standard.imperfect}. Specifically, $w$ is replaced by $w_\urcorner$ which translates anticlockwise along the bumpy boundary. The latter has the same $K$-theory class as its ``smoother'' version $w_\urcorner'$ obtained by shifting $w_\urcorner$ into the bulk (plus some finite-rank terms to make $w_\urcorner'$ unitary), whence we see easily that the pairings of $[w_\urcorner]$ with the cyclic cocycles $\xi_i$ are the same as those for $w$. That the {\rm Exp} map still takes $[P_-]\mapsto k[w_\urcorner]$ follows from Theorem \ref{thm:staircase.LES}.

Let us remark that near the corner, the separation of contributions by states near $F_1$ and near $F_2$ according to the definitions of $\tau_{F_i}$ is not precise, so the quantisation of currents along a face is expected to be exact only when measured far enough from the corner. This is in line with the exact quantisation of boundary currents obtained anywhere along the boundary in a half-plane geometry \cite{PSB}. 

\section{Further generalisations}\label{sec:generalisations}
\subsection{Cornering around irrational slope faces}\label{sec:irrational.computations}
Now suppose the cone $C$ has one/both of its slopes $\alpha_1<\alpha_2$ being \emph{irrational}. As before, we may still define $C^*_r(S)$ for the subsemigroup $S=C\cap\ZZ^2$ generated by the truncations $\hat{U}_\gamma^{\alpha_1,\alpha_2}$ of $U_\gamma, \gamma\in\ZZ^2$, and there is a short exact sequence
$$0\rightarrow\mathcal{I}\rightarrow C^*_r(S)\xrightarrow{\pi} C^*_r(\ZZ^2)\rightarrow 0$$
with $\mathcal{I}$ the commutator ideal \cite{CoDo}. However $S$, and thus $C^*_r(S)$ is no longer finitely-generated, so specifying $\mathcal{I}$ and its $K$-theory in terms of generating operators is not fruitful. Instead, we need some constructions from \cite{ParkThesis, Parkcones, Ji, Jiang}.

First consider the half-planes $y\geq\alpha_1 x$ and $y\leq\alpha_2 x$, and corresponding half-plane Toeplitz algebras $\mathcal{T}^{\alpha_i}$ generated by truncations $\hat{U}\hat{U}_\gamma^{\alpha_i}$ of $U_\gamma, \gamma\in\ZZ^2$ to the respective half-planes. There are
short exact sequences
$$0\rightarrow \mathcal{J}^{\alpha_i}\rightarrow\mathcal{T}^{\alpha_i} \xrightarrow{\pi_i} C^*_r(\ZZ^2)\rightarrow0$$
with $\pi_i:\hat{U}_\gamma^{\alpha_i}\mapsto U_\gamma$, generalising Eq. \eqref{standard.quarter.plane.SES}. The LES for rational $\alpha_i$ is essentially Eq.\ \eqref{standard.quarter.plane.SES} (make an integral basis change from the standard $\ZZ^2$ to exhibit $\mathcal{T}^{\alpha_i}\cong C^*_r(\NN\times\ZZ)$) while for irrational $\alpha_i$, this was computed in Prop.\ 3.2 of \cite{Ji}; the result that we need later is $K_1(\mathcal{J}^{\alpha_i})\cong\ZZ$. 

Next, there are also surjective morphisms $\eta_i:C^*_r(S)\rightarrow \mathcal{T}^{\alpha_i}$ taking $\hat{U}_\gamma^{\alpha_1,\alpha_2}\mapsto \hat{U}_\gamma^{\alpha_i}$, and we can define the pullback
\begin{equation}\label{pullback.Toeplitz}
\xymatrix{
\mathcal{S}^{\alpha_1,\alpha_2}\ar[r]^{\rho_1}\ar[d]_{\rho_2} & \mathcal{T}^{\alpha_1}\ar[d]^{\pi_1}  \\
\mathcal{T}^{\alpha_2}\ar[r]^{\pi_2} & C^*_r(\ZZ^2)
}
\end{equation}
By analysing the Mayer--Vietoris sequence for this pullback (and the computations of $K_\bullet(\mathcal{T}^{\alpha_i})$ in \cite{Ji}), see \S 3 of \cite{Parkcones}, one deduces that $K_1(\mathcal{S}^{\alpha_1,\alpha_2})\cong\ZZ$, and that $K_0(\mathcal{S}^{\alpha_1,\alpha_2})\cong\ZZ^2$ or $\ZZ^3$ depending on whether one/both of $\alpha_1, \alpha_2$ is/are irrational. The class $[1]$ of the identity projection is a generator in either case. 

To relate these $K$-groups of $\mathcal{S}^{\alpha_1,\alpha_2}$ to those of $C^*_r(S)$, we need the short exact sequence (Corollary 2.5 of \cite{Parkcones})
$$0\rightarrow \mathcal{K}(\ell^2(S))\rightarrow C^*_r(S)\xrightarrow{q} \mathcal{S}^{\alpha_1,\alpha_2}\rightarrow 0,$$
with $q(\cdot)=(\eta_1(\cdot),\eta_2(\cdot))$. By exhibiting a Fredholm operator $T\in C^*_r(S)$ of index $-1$ \cite{Parkcones,Jiang}, the connecting index map in $K$-theory for the above sequence is shown to be an isomorphism, so one deduces that $K_1(C^*_r(S))=0$ while $K_0(C^*_r(S))\cong K_0(\mathcal{S}^{\alpha_1,\alpha_2})$ is $\ZZ[1]\oplus \ZZ$ or $\ZZ[1]\oplus\ZZ^2$ as the case may be. When $q$ is restricted to $\mathcal{I}\subset C^*_r(S)$, we obtain a short exact sequence \cite{ParkThesis, Jiang}
$$0\rightarrow \mathcal{K}(\ell^2(S))\rightarrow\mathcal{I}\xrightarrow{q} \mathcal{J}^{\alpha_1}\oplus\mathcal{J}^{\alpha_2}\rightarrow 0,$$
generalising Eq. \eqref{eqn:ideal.modK.rational}. In its long exact sequence,
\begin{equation*}
\xymatrix{
{\overbrace{K_0(\mathcal{K}(\ell^2(S))}^{\ZZ}} \ar[r] & K_0(\mathcal{I}) \ar[r]^{q_*\qquad} & K_0(\mathcal{J}^{\alpha_1})\oplus K_0(\mathcal{J}^{\alpha_2}) \ar[d]^{{\rm Exp}} \\
{\underbrace{K_1(\mathcal{J}^{\alpha_1})\oplus K_1(\mathcal{J}^{\alpha_2})}_{\ZZ[\omega_{\alpha_1}]\oplus\ZZ[\omega_{\alpha_2}]}} \ar[u]^{{\rm Ind}} & {\underbrace{K_1(\mathcal{I})}_{?}} \ar[l] & {\underbrace{K_1(\mathcal{K}(\ell^2(S))}_{0}}\ar[l]_{0}
}
\end{equation*}
the index map is similarly shown to be surjective due to $T-1\in\mathcal{I}$ \cite{Jiang}, and then $K_1(\mathcal{I})\cong\ZZ[w]$ is deduced. Here, the representative $w\in\mathcal{I}$ is a unitary lift of $(\omega_{\alpha_1},\omega_{\alpha_2}^*)$ as in Lemma \ref{lem:ideal.k-theory}.

With these results at hand, we immediately deduce:
\begin{proposition}\label{prop:irrational.exponential.map}
In the LES for $0\rightarrow\mathcal{I}\rightarrow C^*_r(S)\xrightarrow{\pi} C^*_r(\ZZ^2)\rightarrow 0$, where $S\subset\ZZ^2$ is the subsemigroup corresponding to a cone with one/both extremal rays having irrational slope,
\begin{equation}\label{irrational.quarter.plane.LES}
\xymatrix{
K_0(\mathcal{I}) \ar[r] & {\overbrace{K_0(C^*_r(S))}^{\ZZ[1]\oplus\ZZ\, {\rm or}\, \ZZ[1]\oplus\ZZ^2}} \ar[r]^{\pi_*} & {\overbrace{K_0(C^*_r(\ZZ^2))}^{\ZZ[1]\oplus\ZZ[\mathfrak{b}]}} \ar[d]^{{\rm Exp}} \\
K_1(C^*_r(\ZZ^2)) \ar[u]^{{\rm Ind}} & {\underbrace{K_1(C^*_r(S))}_{0}} \ar[l]& {\underbrace{K_1(\mathcal{I})}_{\ZZ[w]}}\ar[l]_{0}
}
\end{equation}
we still have ${\rm Exp}:[\mathfrak{b}]\mapsto[w]$.
\end{proposition}

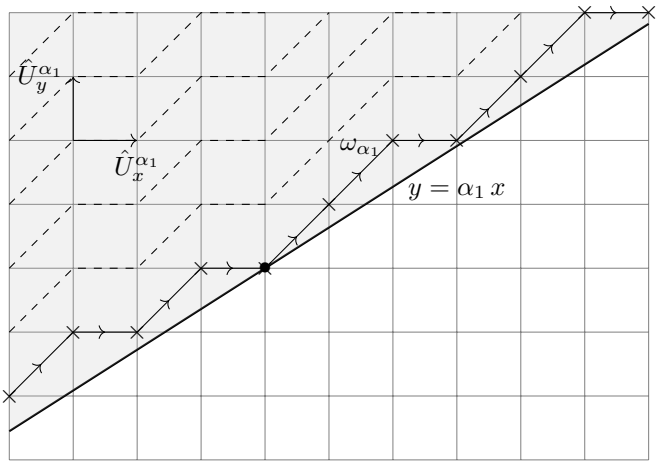
\begin{figure}
\begin{center}
\begin{tikzpicture}[scale=0.85]
\draw[help lines] (-4,-3) grid (6,4);
\draw[thick] (-4,-2.551) -- (6,3.826);
\fill[lightgray, opacity=0.2] (-4,-2.551) -- (6,3.826) -- (6,4) -- (-4,4) -- (-4,-2.551);

\node at (0,0) {$\bullet$};

\node (g) at (0,0) {$\times$};
\node (h) at (1,1) {$\times$};
\node (i) at (2,2) {$\times$};
\node (j) at (3,2) {$\times$};
\node (k) at (4,3) {$\times$};
\node (l) at (5,4) {$\times$};
\node (m) at (6,4) {$\times$};
\node (f) at (-1,0) {$\times$};
\node (e) at (-2,-1) {$\times$};
\node (d) at (-3,-1) {$\times$};
\node (c) at (-4,-2) {$\times$};

\path
   (-4,-2) edge[->-] (-3,-1)
   (-3,-1) edge[->-] (-2,-1)
   (-2,-1) edge[->-]  (-1,0)
   (-1,0) edge[->-]  (0,0)
   (0,0) edge[->-] (1,1)
   (1,1) edge[->-] (2,2)
   (2,2) edge[->-] (3,2)
   (3,2) edge[->-]  (4,3)
   (4,3) edge[->-] (5,4)
   (5,4) edge[->-] (6,4);
 
 \draw[thin, dashed] (-4,-1) -- (-3,0) -- (-2,0) -- (-1,1) -- (0,1) -- (1,2) -- (2,3) -- (3,3) -- (4,4);
  \draw[thin, dashed] (-4,0) -- (-3,1) -- (-2,1) -- (-1,2) -- (0,2) -- (1,3) -- (2,4) -- (3,4);
   \draw[thin, dashed] (-4,1) -- (-3,2) -- (-2,2) -- (-1,3) -- (0,3) -- (1,4);
      \draw[thin, dashed] (-4,2) -- (-3,3) -- (-2,3) -- (-1,4) -- (0,4);
         \draw[thin, dashed] (-4,3) -- (-3,4) -- (-2,4);

\node at (3,1.2) {$y=\alpha_1\,x$};
\node[above] at (1.5,1.6) {$\omega_{\alpha_1}$};

\draw[->] (-3,2) -- (-2,2);
\draw[->] (-3,2) -- (-3,3);
\node[below] at (-2,2) {$\hat{U}^{\alpha_1}_x$};
\node[left] at (-3,3) {$\hat{U}^{\alpha_1}_y$};

\end{tikzpicture}
\caption{Half-plane with irrational slope $0<\alpha_1<1$. The ``boundary projection'' $P_{Y}=1-\hat{U}^{\alpha_1}_y(\hat{U}^{\alpha_1}_y)^*\in\mathcal{J}^{\alpha_1}$ projects onto a set $Y$ of ``approximate boundary points'' as indicated by $\times$. The operator $\omega_{\alpha_1}=P_{Y}(\hat{U}^{\alpha_1}_x+\hat{U}^{\alpha_1}_{(1,1)})P_{Y}$ on $\ell^2(Y)$ effects ``translation along $Y$''. The set $S$ of lattice points in the half-plane can be ``foliated'' by the vertical translates of $Y$ (dashed lines).}\label{fig:irrational.slope}
\end{center}
\end{figure}

At this juncture, we would like to repeat the arguments in Section \ref{sec:topological.cornering.states}, but there are some difficulties. First, if $\alpha_i$ is irrational, we do not have a simple concrete representative\footnote{In the rational case, we can simply choose $\omega_{\alpha_i}$ to be the projection onto the half-plane boundary line $y=\alpha_i x$ composed with the translation $\hat{U}^{\alpha_i}_{a_i}$ along the boundary line.} $\omega_{\alpha_i}$ which generates $K_1(\mathcal{J}^{\alpha_i})\cong\ZZ$. On physical grounds, $\omega_{\alpha_i}$ can be taken to be an ``approximate generating translation along the boundary'', as illustrated in Fig.\ \ref{fig:irrational.slope}. Once we have such $\omega_{\alpha_1}, \omega_{\alpha_2}$, we can take $w=\hat{\omega}_{\alpha_2}+\hat{\omega}_{\alpha_1}^*$ --- the ``anticlockwise translation around the corner'' as in Lemma \ref{lem:ideal.k-theory} --- to represent the generator of $K_1(\mathcal{I})\cong \ZZ$. This veracity of this construction follows from coarse geometry arguments, see \S 2 and \S 5 of \cite{Ludewig-Thiang-cobordism}.

The next difficulty is to construct a trace on $\mathcal{I}$, for which we first need a trace on on $\mathcal{J}^{\alpha_i}$. In the rational case, there is a nice foliation of the half/quarter-plane lattice points by translates of the faces (Fig.\ \ref{fig:rational.angle.concave}), which exhibits an isomorphism $\mathcal{J}^{\alpha_1}\cong \mathcal{K}\otimes C^*_r(\ZZ)$ facilitating the construction of a trace. For the irrational case, the foliation in Fig.\ \ref{fig:irrational.slope} provides a clue in this direction.

Alternatively, there is another way to compute the boundary currents by taking a partition of the half or quarter-plane, and establishing a trace formula for the flow across the partition due to the boundary states. This approach is carried out in detail in \S 4 and \S 6 (especially Theorem 6.1) of \cite{Ludewig-Thiang-cobordism}.

\subsection{Cornering around concave corners}
The faces of a convex cone $C$ also define a closed \emph{concave quarter-plane} $\check{C}$, as illustrated in Fig.\ \ref{fig:rational.angle.concave}. Let $\check{S}=\check{C}\cap\ZZ^2$ be the lattice points in this concave quarter plane (this is no longer a subsemigroup although $S=C\cap\ZZ^2$ still acts on $\check{S}$), then we can truncate $U_\gamma, \gamma\in\ZZ^2$ to operators $\check{U}_\gamma$ acting on $\ell^2(\check{S})$, and generate the $C^*_r$-algebra $C^*_r(\check{S})$ in an analogous way to $C^*_r(S)$. In \cite{Hayashi2}, it was shown that there is again a short exact sequence
\begin{equation}
0\rightarrow \check{\mathcal{I}}\rightarrow C^*_r(\check{S})\xrightarrow{\pi} C^*_r(\ZZ^2)\rightarrow 0,\label{eqn:concave.SES}
\end{equation}
and that the computations of Section \ref{sec:irrational.computations} generalise almost verbatim. It is then easy to deduce that the analogue of Proposition \ref{prop:irrational.exponential.map} still holds: the exponential map in the LES for Eq.\ \eqref{eqn:concave.SES} maps the Bott element $[\mathfrak{b}]$ to the generator $[\check{w}]$ of $K_1(\check{\mathcal{I}})\cong\ZZ$, and a representative of the latter can be taken to be the unitary clockwise translation around the corner of the concave quarter-plane (Fig.\ \ref{fig:rational.angle.concave}).

\subsection{Magnetic translations and quantum Hall effect}\label{sec:magnetic.translations}
For the integer quantum Hall effect, $C^*_r(\ZZ^2)$ is replaced by a twisted version $C^*_r(\ZZ^2,\theta)$ generated by a \emph{projective} regular representation $\ZZ^2\ni \gamma\mapsto T_\gamma\in \mathcal{B}(\ell^2(\ZZ^2))$ with 2-cocycle $\sigma(\gamma_1,\gamma_2)={\rm exp}(-\im\pi \theta \gamma_1\wedge\gamma_2)$ \cite{Bellissard,BSB}. (Dual) \emph{magnetic translations} realise such a representation, and $T_xT_y=e^{2\pi \im \theta} T_yT_x$ for instance. The parameter $\theta\in \RR$ is the strength of a constant magnetic field applied perpendicularly to the 2D sample. One studies Hamiltonians which commute with magnetic translations, and the spectral projections $P_-$ live in some $M_N(C^*_r(\ZZ^2,\theta))$. Now, it is known  \cite{Rieffel} that $C^*_r(\ZZ^2,\theta)$ is a noncommutative torus, whose $K$-theory is identical to that of $C^*_r(\ZZ^2)$ with the Bott projection generator $[\mathfrak{b}]$ replaced by the class of the Rieffel projection $[P_{\rm Rieffel}]$. We may construct the twisted semigroup $C^*$-algebra $C^*_r(S,\theta)$ as the truncated version of $C^*_r(\ZZ^2)$, and obtain (commuting) face projections $P_{F_1}, P_{F_2}$ and the boundary-translation operator $w$ in the kernel of the canonical map $C^*_r(S,\theta)\rightarrow C^*_r(\ZZ^2,\theta)$; the results of Section \ref{sec:quarter.plane.section}-\ref{sec:topological.cornering.states} carry over in an almost identical way to the twisted case, although we omit details of the computation.

\vspace{1em}
{\bf Acknowledgements.} We thank E. Prodan, S. Hayashi, and M. Ludewig for insightful discussions and suggestions. This work was supported by Australian Research Council Discovery Grant DE170100149.

\end{document}